\documentclass[11pt,reqno]{amsart}
\usepackage{amssymb}

\usepackage{enumitem}
\usepackage{amsthm,amsfonts,amssymb,euscript,color,bbm}
\usepackage{comment}
\usepackage{mathtools}
\usepackage{stmaryrd}
\usepackage{mathrsfs}
\usepackage{pifont}
\usepackage{amsfonts,amsmath,amssymb,amsthm,amscd,cancel}
\usepackage{graphicx}
\usepackage[]{geometry}
\usepackage{verbatim}
\usepackage{mathrsfs}
\usepackage{fancyhdr}
\usepackage{footnpag,footmisc}%footnote tongyibianhao
\usepackage[normalem]{ulem}

\usepackage{cite}%[4-8]
\usepackage{hyperref}
\usepackage{engord}% 1st 2nd 3rd
\usepackage[displaymath]{lineno}%line NO.

\theoremstyle{definition}
\newtheorem{theorem}{Theorem}[section]
\newtheorem{lemma}{Lemma}[section]
\newtheorem{proposition}{Proposition}[section]

\newtheorem{remark}{Remark}[section]

\newtheorem*{mth}{Main Theorem}

\allowdisplaybreaks[4]

\DeclareMathAlphabet{\mathsfsl}{OT1}{cmss}{m}{sl}
\numberwithin{equation}{section}

\newcommand{\D}{\mathrm{d}}

\newcommand{\tr}{\mathrm{tr}}

\def\alphab{\underline{\alpha}}
\def\betab{\underline{\beta}}
\def\chib{\underline{\chi}}
\def\chibh{\widehat{\underline{\chi}}}
\def\chih{\widehat{\chi}}
\def\etab{\underline{\eta}}

\def\Lb{\underline{L}}

\def\tr{\mathrm{tr}}
\def\omegab{\underline{\omega}}

\def\tensor{\widehat{\otimes}}

\def\ub{\underline{u}}
\def\Cb{\underline{C}}

\def\Lh{\widehat{L}}
\def\Lbh{\widehat{\underline{L}}}

\newcommand{\Db}{\underline{D}}
\newcommand{\Dh}{\widehat{D}}
\newcommand{\Dbh}{\widehat{\underline{D}}}

\def\nablas{\mbox{$\nabla \mkern -13mu /$ }}
\def\Deltas{\mbox{$\Delta \mkern -13mu /$ }}

\def\divs{\mbox{$\mathrm{div} \mkern -13mu /$ }}
\def\curls{\mbox{$\mathrm{curl} \mkern -13mu /$ }}

\def\ds{\mbox{$\nabla \mkern -13mu /$ }}
\def\gs{\mbox{$g \mkern -9mu /$}}
\def\epsilons{\mbox{$\epsilon \mkern -9mu /$}}

\begin{document}

\title{A construction of collapsing spacetimes in vacuum}

\author[Junbin Li]{Junbin Li}
\address{Department of Mathematics, Sun Yat-sen University, Guangzhou, China}
\email{lijunbin@mail.sysu.edu.cn}

\author[He Mei]{He Mei}
\address{Department of Mathematics, Sun Yat-sen University, Guangzhou, China}

\email{meihe@mail2.sysu.edu.cn}

\begin{abstract}
In this paper, we construct a class of collapsing spacetimes in vacuum without any symmetries. The spacetime contains a black hole region which is bounded from the past by the future event horizon. It possesses a Cauchy hypersurface with trivial topology which is located outside the black hole region. Based on existing techniques in the literature, the spacetime can in principle be constructed to be past geodesically complete and asymptotic to Minkowski space. The construction is based on a semi-global existence result of the vacuum Einstein equations built on a modified version of the a priori estimates that were originally established by Christodoulou in his work on the formation of trapped surface, and a gluing construction carried out inside the black hole. In particular, the full detail of the a priori estimates needed for the existence is provided, which can be regarded as a simplification of Christodoulou's original argument.
\end{abstract}

\maketitle

\tableofcontents

\setcounter{tocdepth}{1}

%\parskip=\baselineskip
%\tableofcontents

\section{Introduction}

\subsection{Previous works}
The problem of gravitational collapse is formulated in terms of the initial value problem of the Einstein equations
$$\mathbf{Ric}_{\alpha\beta}-\frac{1}{2}\mathbf{R}g_{\alpha\beta}=8\pi \mathbf{T}_{\alpha\beta}$$
which are the fundamental system of equations in general relativity. A spacetime of gravitational collapse has a complete regular past and black holes in the future. In general its Penrose diagram can be depicted as in Figure \ref{fig:collapse}, where $\mathcal{I}^{\pm}$ are the future and past null infinity, $i^{\pm}$ are the future and past timelike infinity, $i^0$ is the spatial infinity and $\mathcal{H}^+$ is the future event horizon. The grey region is called the domain of outer communication, the region outside the black hole. We say a spacetime has a black hole if it has a \emph{complete future null infinity} and a future event horizon, that is, the whole grey region can be identified in the spacetime.\footnote{The story inside the black holes is however another interesting topic and readers can refer to a recent work \cite{D-L} and the references therein.} To understand the spacetimes of gravitational collapse, the first step would be to provide certain examples of such spacetimes, which is still rather difficult in general relativity. The first example is due to Oppenheimer-Snyder \cite{O-S}, who considered the spherically symmetric homogeneous dust model and the infinite redshift effect was discovered for the first time, even though the concepts of black hole and event horizon were not introduced explicitly.

\begin{figure}[htbp]
\centering
\includegraphics [width=2 in]{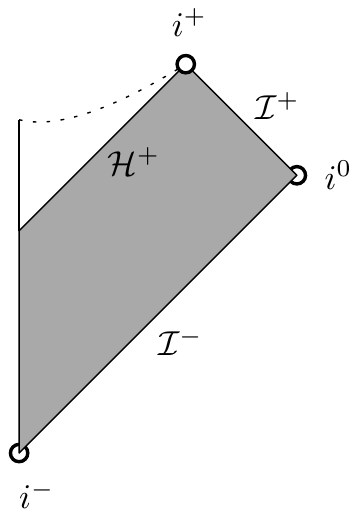}  
\caption{}
 \label{fig:collapse} 
\end{figure}

A well-understood model is the spherically symmetric massless scalar field model, where 
$$\mathbf{T}_{\alpha\beta}=\partial_\alpha\phi\partial_\beta\phi-\frac{1}{2}g_{\alpha\beta}\partial_\mu\phi\partial^\mu\phi$$
for a real scalar function $\phi$. It had been extensively studied by Christodoulou in a series of papers. He considered the characteristic initial value problem where the initial data is prescribed on a null cone emanating from a point,  and showed in \cite{Chr91} that under a sharp lower bound condition on the initial data, the spacetime must collapse into a black hole with strictly spacelike singular boundary. Dafermos had constructed in \cite{D09} a family of solutions with complete regular past such that the lower bound condition holds on some outgoing null cone and consequently the spacetimes collapse into a black hole. Therefore we have examples of collapsing spacetimes in this context. In fact, far more than providing examples, Christodoulou had showed in \cite{Chr99} that generic spacetimes with non-complete future (in this context) do have a \emph{complete future null infinity} and hence a future event horizon, verifying the weak cosmic censorship conjecture for spherically symmetric solutions of the Einstein equations coupled with massless scalar field.

When $\mathbf{T}_{\alpha\beta}$ is set to be zero, the system of equations is called the vacuum Einstein equations. It would be natural to ask whether a black hole can form in pure general relativity. Because spherically symmetric vacuum solutions are static Schwarzschild spacetimes and hence admit no dynamical freedoms, we should consider solutions outside spherical symmetry and this makes the problem rather difficult. Another difficulty is that it should be formulated in terms of the large data problem, otherwise a black hole cannot form. It was a longstanding problem and was answered by Christodoulou in his breakthrough \cite{Chr08}. Before stating his theorem, let us first introduce the concept of trapped surface. A trapped surface is a two dimensional spacelike surface with its mean curvatures relative to both future null normals being negative. By Penrose singularity theorem, the existence of a closed trapped surface in a globally hyperbolic spacetime with non-compact Cauchy hypersurface implies future null incompleteness of the spacetime. Moreover, if the future null infinity $\mathcal{I}^+$ is well-defined, then signals from the closed trapped surface cannot be sent to $\mathcal{I}^+$. Thus if there is a closed trapped surface, then it is contained in a black hole region, which in particular nonempty, and hence an event horizon exists if the future null infinity is complete.

The existence of black hole and event horizon is difficult to detect because it is about the global behavior of the spacetime. It is more  convenient to study the existence of a closed trapped surface. So in local terms, we may consider a spacetime of gravitational collapse to be a spacetime with the property that it has a Cauchy hypersurface $\Sigma$ without any closed trapped surfaces and it contains a closed trapped surface to the future of $\Sigma$. Then the theorem of Christodoulou can be stated in a qualitative way as follows.
\begin{theorem}[Christodoulou, \cite{Chr08}]
There exist characteristic initial data sets on a null cone emanating from a point, without any closed trapped surfaces, such that their maximal future developments of the vacuum Einstein equations have a closed trapped surface.
\end{theorem}

\begin{figure}[htbp]
\centering
\begin{minipage}[t]{0.48\textwidth}
\centering
\includegraphics[width=2 in]{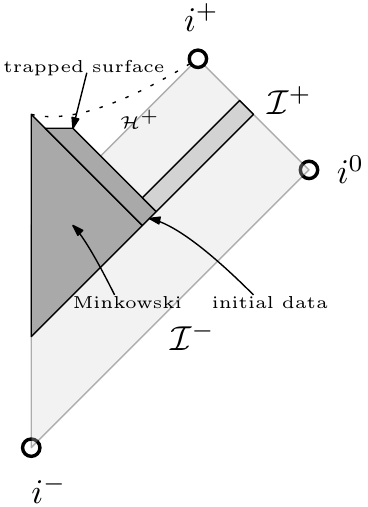}  
\caption{}
 \label{fig:chr} 
\end{minipage}
\begin{minipage}[t]{0.48\textwidth}
\centering
\includegraphics[width=2.3 in]{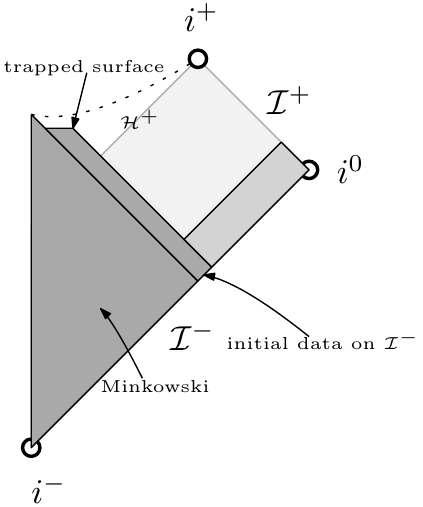}  
\caption{}
 \label{fig:chr1} 
\end{minipage}
\end{figure}

The solution Christodoulou solved is the dark grey region in Figure \ref{fig:chr}. The initial data can be extended to the whole outgoing null cone, which is then asymptotically flat. From a result in \cite{L-Z} of Zhu and the first author, the spacetime contains a piece of the future null infinity $\mathcal{I}^+$, which is the medium grey region in Figure \ref{fig:chr}\footnote{The light grey regions in Figures \ref{fig:chr}--\ref{fig:limei2} represent the whole domain of outer communication.}. On the other hand, the initial data can be prescribed on a piece of the past null infinity $\mathcal{I}^-$ and the solution then includes the dark grey region in Figure \ref{fig:chr1}. Now it is natural to ask if we can include the spatial infinity $i^0$ in our collapsing spacetimes, that is, to consider initial value problem with asymptotically flat Cauchy data. Christodoulou had proposed that if we prescribe the initial data on the whole past null infinity $\mathcal{I}^-$, then the solution would contain the medium grey region in Figure \ref{fig:chr1}, and in particular contain the spatial infinity $i^0$ and a piece of future null infinity which is complete to the past. In this case,  more than possessing a Cauchy hypersurface without any closed trapped surfaces, the spacetime is indeed geodesically complete to the past. For the same purpose, but by a different approach, Yu and the first author were able to show:
\begin{theorem}[Li-Yu, \cite{L-Y}] There exist complete and asymptotically flat Cauchy data sets without any closed trapped surfaces such that their maximal future development contain a closed trapped surface.
\end{theorem}

The method of proving this theorem is to use the local deformation techniques developed by Corvino-Schoen \cite{C00,C-S} (also see Chrusciel-Delay \cite{C-D}) to attach a piece of constant time slice of Kerr spacetime exterior to Christodoulou's collapsing spacetime. Christodoulou's collapsing spacetime is the dark grey region in Figure \ref{fig:liyu} and the thick curve represents the Cauchy data we produce. The Cauchy data then have trivial topology and they are isometric to the constant time slices of Kerr spacetime outside a compact region. This in particular implies that the  maximal  future (and past) developments are also isometric to Kerr spacetimes in a neighborhood of the spatial infinity $i^0$. Moreover, based on existing techniques but lengthy arguments, it should be able to be proved that the maximal past development is past geodesically complete and asymptotic to Minkowski space.

It is then natural to ask if we can reach the future timelike infinity $i^+$, that is, if we can fill in the remaining part of the light grey region in Figures \ref{fig:chr}, \ref{fig:chr1}, \ref{fig:liyu}. If this is the case, then we would have a spacetime of gravitational collapse in a global sense. We would also be able to identify the future event horizon in this spacetime. A direct way is to solve the Einstein equations globally in the domain of outer communication, that is, to solve the weak cosmic censorship conjecture, which is of course far out of reach nowadays. Based on the work \cite{L-Y} of Yu and the first author, we can still reach the future timelike infinity $i^+$ if the stability of Schwarzschild (in the Kerr family) is true, which may be solved in the coming years but still requires lots of works. We remark that in spherical symmetry, the situation is much more clear: the existence of a closed trapped surface would imply the completeness of the future null infinity for suitable matter fields (see \cite{D05}).

\begin{figure}[htbp]
\centering
\begin{minipage}[t]{0.48\textwidth}
\centering
\includegraphics [width=2.1 in]{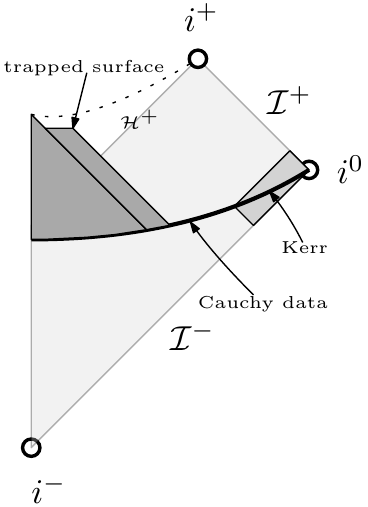}  
\caption{}
 \label{fig:liyu} 
\end{minipage}
\begin{minipage}[t]{0.48\textwidth}
\centering
\includegraphics [width=2 in]{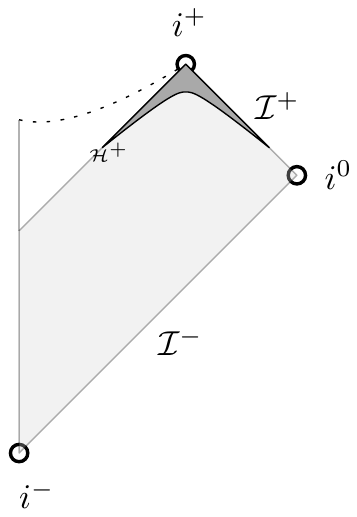}  
\caption{}
 \label{fig:dhr} 
\end{minipage}
\end{figure}

Nevertheless, we may still have some indirect way to exhibit examples of collapsing spacetimes. This would still help to improve our understanding of gravitational collapse. In a recent work \cite{D-H-R}, Dafermos-Hozegel-Rodnianski were able to construct vacuum spacetimes of evolving black holes, which means that the vacuum spacetimes without any stationary Killing fields are constructed up to the future timelike infinity $i^+$, and hence a future event horizon can be identified. Their construction is depicted in the dark grey region of Figure \ref{fig:dhr}. The method is by solving the Einstein equations backward with initial data given on the future event horizon $\mathcal{H}^+$ and the future null infinity $\mathcal{I}^+$. The solution then exists in a neighborhood of $i^+$. The behavior to the past is not concerned in their work.

\subsection{The Main Theorem}

In this paper, we are able to construct an example of a collapsing spacetime in global sense in an indirect way. The main theorem of this paper is the following.
\begin{mth}
There exists a solution $(M,g)$ of the vacuum Einstein equations satisfying the following conditions:
\begin{itemize}
\item $M$ contains a closed trapped surface and has a complete future null infinity. Then $M$ is future null geodesically incomplete and has a black hole region $\mathcal{B}$ containing the closed trapped surface. 
\item $M$ is the future development of an asymptotically flat Cauchy data $\Sigma$ with trivial topology. Moreover, $\Sigma$ contains no points contained in the black hole region, i.e., $\Sigma\cap\mathcal{B}=\varnothing$. In particular, $\Sigma$ contains no closed trapped surfaces.
\end{itemize}
\end{mth}
\begin{figure}[htbp]
\centering
\includegraphics [width=2 in]{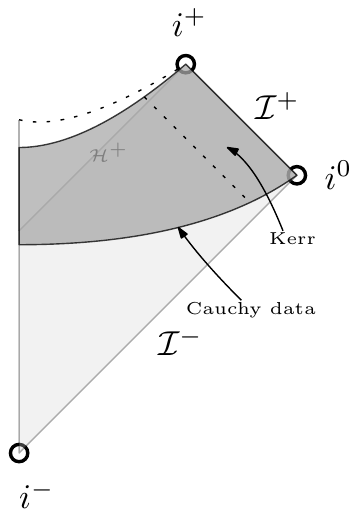}  
\caption{}
 \label{fig:limei} 
\end{figure}

The spacetime $(M,g)$ we construct is depicted in the darker grey region of Figure \ref{fig:limei}. We make several remarks on the Main Theorem and will discuss and begin the proof in the next section. First of all, the spacetimes we construct have a special property: For any $m_0>0$ and $\varepsilon>0$, $(M,g)$ is constructed in the way that, there is a region $\Omega\subset\Sigma$ diffeomorphic to a 3-ball such that outside the future domain of dependence of $\Omega$, $(M,g)$ is isometric to the Kerr spacetime with mass parameter $m$ and angular momentum vector $\mathbf{a}$, where 
$$|m-m_0|+|\mathbf{a}|<\varepsilon.$$
This property is similar to the construction in \cite{L-Y}, by which the construction in this paper is mainly inspired. This is also in some sense very similar to spherically symmetric collapsing spacetimes, which are isometric to the Schwarzschild spacetimes in vacuum region. %Combined with the construction in \cite{D-H-R}, $(M,g)$ can be made not exactly isometric to but exponentially asymptotic to Kerr spacetimes. 
The second thing is that, since the black hole region in the solution can be identified, we are able to state and prove that the Cauchy hypersurface $\Sigma$ satisfies the property that $\Sigma\cap\mathcal{B}=\varnothing$, which is stronger than the property that $\Sigma$ contains no closed trapped surfaces. Similar to the construction in \cite{L-Y}, based on the existing techniques, the maximal past development of  $\Sigma$ should be past geodesically complete and asymptotic to Minkowski space. Therefore, the spacetimes we construct can be viewed as the spacetimes of the formation of Kerr black holes from complete regular and dispersive past.

\subsection{Comments on the proof}

\subsubsection{Main steps of the construction}

\begin{figure}[htbp]
\centering
\begin{minipage}[t]{0.48\textwidth}
\centering
\includegraphics[width=2 in]{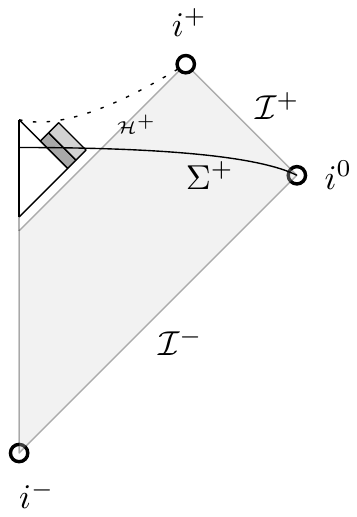}  
\caption{}
 \label{fig:limei1} 
\end{minipage}
\begin{minipage}[t]{0.48\textwidth}
\centering
\includegraphics[width=2 in]{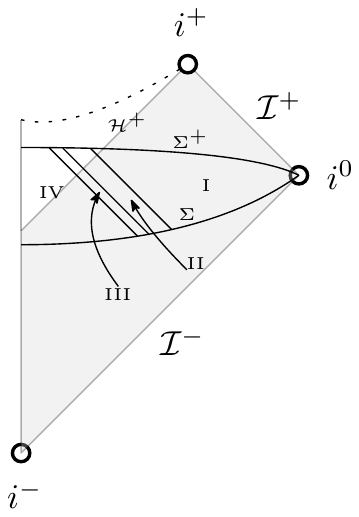}  
\caption{}
 \label{fig:limei2} 
\end{minipage}
\end{figure}

Let us first explain briefly the main steps of the construction.

\begin{itemize}
\item {\bf Step 1.}  Construction of a spacetime region inside the black hole.
\end{itemize}

The first step is to construct the dark and medium grey regions, together with the triangular white Minkowskian region on their left in Figure \ref{fig:limei1}. The dark grey region is the \emph{short pulse} region with strong incoming gravitational wave. By imposing an \emph{isotropic condition}, the medium grey region can be proved to be close to a Schwarzschild spacetime. It is a region after the incoming wave. These regions are constructed as a solution of the vacuum Einstein equations as a characteristic initial value problem as constructed in \cite{Chr08} and \cite{L-Y}. Different to \cite{L-Y}, only the part in the black hole is needed in this step and hence the characteristic initial data can be prescribed in the way that it is already inside the black hole.

\begin{itemize}
\item {\bf Step 2.} Gluing construction inside the black hole.
\end{itemize}

This is the key step, the main novelty of this paper from technical point of view. This step is to construct a Cauchy surface as depicted also in Figure \ref{fig:limei1}. The Cauchy surface $\Sigma^+$ is isometric to a spacelike slice in a Kerr spacetime outside the medium grey region constructed in the previous step. This can be done by perturbing the data induced on this surface in the medium grey region using the gluing techniques developed by Corvino-Schoen \cite{C00,C-S} (also see Chrusciel-Delay \cite{C-D}). This Kerr spacetime in particular has small angular momentum per unit mass. This is similar to the construction in \cite{L-Y}. The difference is that the gluing is done outside the black hole in \cite{L-Y}, while the construction in this article is done \emph{inside} the black hole. 

\begin{itemize}
\item {\bf Step 3.} Construction of the spacetime outside the black hole.
\end{itemize}

From the previous step, it is clear that the maximal future development of $\Sigma^+$ includes the future timelike infinity $i^+$ whose past neighborhood is isometric to a Kerr spacetime. However, $\Sigma^+$ contains at least one closed trapped surface and hence it intersects the black hole region. In order to obtain another Cauchy surface $\Sigma$ that lies outside the black hole region, we need to solve the vacuum Einstein equations from $\Sigma^+$ \emph{to the past}. The solution is depicted in the regions $I$, $II$, $III$ and $IV$ in Figure \ref{fig:limei2}. These four regions are solved sequentially. Let us denote by $\Sigma_{I}^+$, $\Sigma_{II}^+$, $\Sigma_{III}^+$ and $\Sigma_{IV}^+$ the part of $\Sigma^+$ lying in the regions $I$, $II$, $III$ and $IV$ respectively. Let us also denote by $\Cb_{I,II}$, $\Cb_{II,III}$ and $\Cb_{III,IV}$ the null cones separating regions $I$ and $II$, $II$ and $III$, $III$ and $IV$ respectively.

Then first of all, it is clear that region $I$ is isometric to a Kerr spacetime that has small angular momentum per unit mass. Then it can be viewed as a spacetime region that is close to a Schwarzschild spacetime and $\Cb_{I,II}$ is close to a null cone in Schwarzschild spacetime. Also, the estimate of Corvino-Schoen technique tells that the data induced on $\Sigma_{II}^+$ is still close to a spacelike slice in Schwarzschild spacetime (Note that the data has been changed after deformation.). Then region $II$ can be solved by Cauchy stability and is closed to a Schwarzschild spacetime.

To solve region $III$, we start from the initial data induced on $\Cb_{II,III}$ which is close to being Schwarzschildean, and the initial data induced on $\Sigma^+_{III}$, which has not been changed by deformation. Then the data of this latter part satisfies exactly the same estimates in the \emph{short pulse} region as in \cite{Chr08}.  However, the data induced on $\Cb_{II,III}$ does not necessarily satisfy the same estimates. This is simply because the data induced on $\Cb_{II,III}$ is only close to being Schwarzschildean but not necessarily exactly Schwarzschildean. Therefore, to construct region $III$, we need to establish the a priori estimates slightly different from those in \cite{Chr08} and this is the main part of Step 3. We will make more comments about this in the next subsection.

At last, we will show that $\Cb_{III,IV}$, is close to a null cone in Minkowski space using an argument similar to that in \cite{L-Y}, then region $IV$ can also be solved by Cauchy stability because the data induced on $\Sigma_{IV}^+$ is exactly Minkowskian.

\subsubsection{More on solving region $III$}\label{moreonregionIII} Let us make more comments on the a priori estimates needed to solve region $III$ mentioned above. We will not explain the notations which can be found in Section \ref{aprioriestimates}. Recall that the a priori estimates introduced by Christodoulou in \cite{Chr08}, tied to a small positive number $\delta$,  tell us that various geometric quantities in the solution are bounded by $\delta$ to some specific different powers. The bounds are designed such that the a priori estimates can be closed under a bootstrap argument under these bounds and at the same time a closed trapped surface has a chance to form. For example, at the level of $L^\infty$ norm of the curvature, the following bounds are satisfied:
\begin{align}\label{hierarchyChr}
|\alpha|\le C\delta^{-\frac{3}{2}}, \ |\beta|\le C\delta^{-\frac{1}{2}},\  |\rho|,|\sigma|\le C,\  |\betab|\le C\delta, \ |\alphab|\le C\delta^{\frac{3}{2}}.
\end{align}
Christodoulou called this display the \emph{short pulse hierarchy}. However, the data induced on $\Cb_{II,III}$ will be proved to be close to that in the Schwarzschild spacetime, or more precisely, at the level of $L^\infty$ norm of the curvature, the following bounds are satisfied on $\Cb_{II,III}$:
\begin{align}\label{hierarchyclosetoS}
|\alpha|,|\beta|,|\rho-\rho_{\text{Sch}}|,|\sigma|,|\betab|,|\alphab|\le C\delta^{\frac{1}{2}},
\end{align}
where $\rho_{\text{Sch}}$ is the value of $\rho$ in the Schwarzschild spacetime. By comparing \eqref{hierarchyChr} and \eqref{hierarchyclosetoS}, it can be seen that \eqref{hierarchyChr} can never be proved to hold in region $III$. The best estimates that can be expected would be the ``union'' of \eqref{hierarchyChr} and \eqref{hierarchyclosetoS}:
\begin{align}\label{hierarchy}
|\alpha|\le C\delta^{-\frac{3}{2}}, \ |\beta|\le C\delta^{-\frac{1}{2}},\  |\rho|,|\sigma|\le C,\  |\betab|, |\alphab|\le C\delta^{\frac{1}{2}},
\end{align}
where $\betab$ and $\alphab$ have worse bounds than in \eqref{hierarchyChr}. We may call this display a \emph{modified short pulse hierarchy}. The existence result in \cite{Chr08} cannot be directly applied here. A generalized existence result in \cite{L-R17} can be applied in the current setting, but it still requires additional works to derive the estimates \eqref{hierarchy}, which are needed in solving region $IV$ but do not directly follows from the a priori estimates derived in \cite{L-R17}. 

For this reason, and also for the sake of self-containedness, we choose to directly derive new a priori estimates built on the hierarchy \eqref{hierarchy} in detail, which are sufficient for the existence of the solution in region $III$. The method to derive a priori estimates in this paper is similar to that in \cite{Chr08}, but we will encounter different \emph{borderline terms}. Besides this, we are able to  write down the proof in the way that the \emph{elliptic estimates are avoided}\footnote{An existence result without elliptic estimates in a more general setting was given recently by An in \cite{An}. The argument presented in this article is originated and modified from that in the first author's thesis \cite{Lithe}.}. So the derivation of a priori estimates in this paper can be regarded as a \emph{simplification} of the argument in \cite{Chr08}\footnote{More precisely, Chapter 3 to Chapter 16.1 in \cite{Chr08}.}, which is one another motivation to write down the estimates in detail.

\subsubsection{Extending the solution to the past null infinity}
We have mentioned above that the maximal past development can in principle be proved to be past geodesically complete and asymptotic to Minkowski space based on existing techniques. We will not pursue this goal in this paper but it can be made clear. In fact, regions $I$ and $II$ can be extended to the past null infinity using the techniques developed in the works of stability of Minkowski space, like \cite{Ch-K,K-N}. This is essentially the problem of the stability of external region. A more related reference is \cite{C-N}, in which the stability of the external regions of Kerr spacetimes was proved in detail.  To extend region $III$ to the past null infinity, in addition to the estimates of modified short pulse hierarchy mentioned before, one needs to estimate at the same time the decay of various quantities to the past null infinity.  This has been done in Christodoulou's work \cite{Chr08} with his original short pulse hierarchy. There should not be essential differences here. However, in the current paper, we solve the solution from finite region to the past null infinity, while Christodoulou solved the solution beginning from the past null infinity. This requires some additional works in choosing suitable foliations of the solution, see also \cite{Ch-K,K-N}. In \cite{L-Z}, we have treated this carefully in a similar semi-global setting. Extending region $IV$ to the past (null and timelike) infinity is essentially the stability of the interior region of Minkowski space, with initial data given on an asymptotically flat null cone, emanating from a point or the spherical boundary of a $3$-dimensional disk. There is no doubt that this can be done in view of all existing works on the stability of Minkowski space, despite the lack of explicit details in the literature.

\subsection{Acknowledgement}
Both authors are supported by NSFC 11822112, 11501582 and 11521101. They would like to thank the anonymous referees for valuable comments on the manuscript.

\section{The proof of the Main Theorem}\label{sec:proof}

We start the proof of the Main Theorem. 

\subsection{Step 1: Construction of a spacetime region inside the black hole}\label{section:step1}

In this subsection we carry out Step 1 of the construction. Let us consider a null cone denoted by $C_{u_0}$ emanating from a point. The initial data set on $C_{u_0}$ is specified in the following way:
\begin{itemize}[leftmargin=0.5cm]
 \item Beginning from the vertex of $C_{u_0}$, the data is exactly Minkowskian on $C_{u_0}$ before a spherical section that is isometric to a standard round sphere of radius $-u_0$. Let $\ub$ be the affine parameter of the null geodesic generators $C_{u_0}$ and $S_{\ub,u_0}$ be the spherical sections which are the level sets of $\ub$ on $C_{u_0}$. Moreover, $\ub$ is chosen in the way that the radius of the section $S_{\ub,u_0}$  is $\ub-u_0$ for $u_0\le \ub\le0$.
\item The data for $\ub\ge0$ consists of the conformal metric $\widehat{\gs}=\widehat{\gs}(\ub,\vartheta)$ on each spherical section. This means that the full metric on each spherical section is given by  $\gs=\phi^2\widehat{\gs}$ where $\phi=\phi(\ub,\vartheta)>0$ is determined by $\widehat{\gs}$. The conformal metric $\widehat{\gs}$ is specified in the following way:

Let $\{(U_1,(\vartheta^A_1)),(U_2,(\vartheta^A_2))\}_{A=1,2}$ be the two stereographic charts on $S_{0,u_0}$. Thus, the round metric $\gs|_{S_{0,u_0}}$ is expressed as $(\gs|_{S_{0,u_0}})_{AB}(\vartheta)=\frac{|u_0|^2}{(1+\frac{1}{4}|\vartheta|^2)^2}\delta_{AB}$ with $\vartheta=\vartheta_1$ or $\vartheta_2$ and $|\vartheta|^2=|\vartheta^1|^2+|\vartheta^2|^2$. Extend the both coordinate systems to the whole $C_{u_0}$, at least for $\ub\ge0$, in the way that $\vartheta^A$ is constant along null geodesic generators of $C_{u_0}$. Let us impose an additional condition on the conformal metric $\widehat{\gs}$ that, written in coordinate systems,
\begin{equation*}
\det(\widehat{\gs}(\ub,\vartheta)_{AB})=\det(\gs(0,\vartheta)_{AB})=\frac{|u_0|^4}{(1+\frac{1}{4}|\vartheta|^2)^4}.
\end{equation*}
Then let us write
\begin{equation*}
\widehat{\gs}(\ub,\vartheta)_{AB}=\frac{|u_0|^2}{(1+\frac{1}{4}|\vartheta|^2)^2}m_{AB}(\ub,\vartheta)=\frac{|u_0|^2}{(1+\frac{1}{4}|\vartheta|^2)^2}\exp\psi_{AB}(\ub,\vartheta),
\end{equation*}
where $m_{AB}$ takes values in the set of positive definite symmetric matrices with determinant $1$ and $\psi_{AB}$ takes value in the set of symmetric trace-free matrices. After this reduction, to prescribe initial data, we only need to specify a function\footnote{In fact, we specify a pair of functions satisfying a compatibility condition on a sphere.} $\psi=\psi(\ub,\vartheta)$ taking values in $\widehat{S}_2$ 
where $\widehat{S}_2$ denotes the set of $2\times 2$ symmetric trace-free matrices.

Now let us choose a smooth compactly supported $\widehat{S}_2$-valued function $\psi_0 \in C^\infty_c((0,1) \times S_{0,u_0})$. Given a small parameter $\delta>0$, then for $0\le\ub\le\delta$, the function $\psi$ is given by setting
\begin{equation}\label{shortpulse}
\psi(\ub,\vartheta)=\frac{\delta^{\frac{1}{2}}}{|u_0|}\psi_0(\frac{\ub}{\delta},\vartheta),
\end{equation}
which was called \emph{short pulse ansatz} by Christodoulou and $\psi_0$ was called the \emph{seed data}. For $\ub\ge\delta$, we set $\psi(\ub,\vartheta)\equiv\psi(\delta,\vartheta)$. Similar to the previous work \cite{L-Y}, we especially choose the seed data $\psi_0$ such that
\begin{align}\label{integrate=m0}
\int_0^1\left|\frac{\partial \psi_0}{\partial s}\right|^2\D s=16m_0
\end{align}
for some $m_0>0$.

\end{itemize}

 Let us denote $\ub,u$ be two optical functions in the spacetime such that, $\Cb_{\ub}$, the level sets of $\ub$, be the incoming null cone emanating from the spherical section $S_{\ub,u_0}$, and $C_u$, the level sets of $u$, be the outgoing null cone emanating from the spherical sections of $\Cb_0$ with radius $|u|$. Denote $S_{\ub,u}=C_u\cap \Cb_{\ub}$. In particular, $\ub$ and $u$ increase to the future.

%To state the result in this step, we also need to define a Schwarzschild metric $g_{m_0}$ with mass $m_0$ in the way that we consider it as a solution with initial data given on $\Cb_{\delta}$ and the part $\ub\ge\delta$ of $C_{u_0}$ . The initial data set consists of the following quantities: a metric $\gs_{m_0}(\delta,u_0)$ on $S_{\delta,u_0}$ which is the round metric of the sphere with radius $|u_0|$, the torsion $\zeta_{m_0}(\delta,u_0) \equiv 0$ on $S_{\delta,u_0}$, two null expansions $\tr\chi_{m_0}(\delta,u_0)\equiv\frac{2}{|u_0|}-\frac{4m_0}{|u_0|^2}$ and $\tr\chib_{m_0}(\delta,u_0) \equiv -\frac{2}{|u_0|}$ on $S_{\delta,u_0}$, the lapse function $\Omega_{m_0} \equiv 1$ on $C_{u_0}\cup\Cb_{\delta}$ and two shears $\chih_{m_0} \equiv 0$ on $C_{u_0}$ and $\chibh_{m_0} \equiv 0$ on $\Cb_\delta$. By Birkhoff theorem and the uniqueness of the characteristic initial data problem, we conclude that the development $g_{m_0}$ is locally isometric to Schwarzschild with mass $m_0$. Combining all the above things together, we have the following conclusion: (see figure \ref{fig1} and \cite{Chr}\cite{Li}\cite{L-Y})

Let us choose the background Schwarzschild metric $g_{m_0}$ with mass $m_0$ such that in the region $\delta\le\ub\le\delta+\varepsilon$, $u_0\le u\le u_1$, $g_{m_0}$ is written in the form
\begin{equation}\label{gm0}g_{m_0}=-4\Omega^2\D\ub\D u+r^2\D\sigma_{\mathbb{S}^2}.\end{equation}
The function $\Omega$ and $r$ are determined by their values on $\Cb_{\delta}$ and $C_{u_0}$ (the part $\delta\le\ub\le\delta+\varepsilon_0$): $\Omega$ is set to be constant $1$ on $\Cb_{\delta}\cup C_{u_0}$, and $r=|u|$ on $\Cb_\delta$, $r=|u_0|+\ub\left(1-\frac{2m_0}{|u_0|}\right)$ on $C_{u_0}$. In particular, the optical functions $\ub, u$ are still optical functions relative to $g_{m_0}$, and are affine parameters of the null generators of $\Cb_{\delta}$ and $C_{u_0}$ relative to $g_{m_0}$. The spheres $S_{\ub,u}$ are orbit spheres of the isometric group of $g_{m_0}$.
 
 Then the spacetime is constructed according to the following theorem.

\begin{theorem}\label{step1}
Let $\delta>0$ be a small number, $m_0>0$, $u_0<u_1<0$ and  $k$ be a large integer. Then there exists an $\varepsilon_0>0$ (only depending on $m_0$), such that if $\delta$ is sufficiently small, the solution $g$ of the vacuum Einstein equations exists for $0\le\ub\le\delta+\varepsilon_0$, $u_0\le u\le u_1$. In addition, the solution $g$ in $\delta\le\ub\le\delta+\varepsilon_0$, $u_0\le u\le u_1$ is $\delta^{\frac{1}{2}}$-close to the Schwarzschild metric $g_{m_0}$ with mass $m_0$ also defined in the same region in the $C^k$ topology.
\end{theorem}
%\begin{remark}
%The closeness of $g$ to $g_{m_0}$ is also uniformly in $u_0$ so that they are close all the way up to the past null infinity, see \cite{Li}. However, we do not need this statement in the present article.
%\end{remark}

This theorem can in fact be found in \cite{L-Y}, in which the proof was largely based on \cite{Chr08}. According to Section \ref{moreonregionIII}, we still need to derive new a priori estimates, generalizing and simplifying those in \cite{Chr08}, to prove the existence of the solution in region $III$. So we choose not to apply results in \cite{Chr08} and \cite{L-Y} but to prove this theorem directly. The derivation of the a priori estimates and the proof of Theorem \ref{step1} are included in Section \ref{aprioriestimates}. We may set in particular $|u_0|<2m_0$, then the solution in $0\le\ub\le\delta, u_0\le u\le u_1$ corresponds to the dark grey region and the solution in $\delta\le\ub\le\delta+\varepsilon_0, u_0\le u\le u_1$ corresponds to the medium grey region in Figure \ref{fig:limei1}\footnote{In fact, we only need to pick $u_1$ such that $|u_1|<2m_0$. Then not the whole but part of the solution will correspond to the dark and medium grey regions.}. And the solution in the medium grey region is close to a Schwarzschild metric inside the Schwarzschild black hole.
%\begin{remark}The statement above in fact implies the existence of a closed trapped surface. This is because the sphere $\ub=\delta, u=-1$ is close to the corresponding sphere $r=1$ in a Schwarzschild spacetime with mass $m_0>1/2$, which is a closed trapped surface. 
%\end{remark}

\subsection{Step 2: Gluing construction inside the black hole} In this subsection we carry out Step 2 of the construction. We need to choose a Cauchy hypersurface $\Sigma^+$ in the spacetime region constructed in the previous step and deform it such that it can be extended to spatial infinity. 

Recall that we have two metrics, the solution $g$ and the Schwarzschild metric $g_{m_0}$ defined in the region $\delta\le\ub\le\delta+\varepsilon_0, u_0\le u\le u_1$, which are $\delta^{\frac{1}{2}}$-close to each other. We first choose a function $\Sigma^+_{II}(\ub)$ such that the hypersurface $u=\Sigma_{II}^+(\ub),\delta\le\ub\le\delta+\varepsilon_0$, is exactly the hypersurface $r_{g_{m_0}}=r_0<2m_0$ relative to the Schwarzschild metric $g_{m_0}$. We denote this hypersurface (excluding its boundary) by $H$ and the initial data induced on it by $g$ by $(\bar{g},\bar{k})$. A key property of $H$ is that the $H$ will not shrink as $\delta\to0$ because $\varepsilon_0$ does not depend on $\delta$. The closeness to Schwarzschild metric implies that, given large integer $k$ ({\bf which may be different in the whole construction}),
$$\|\bar{g}-\bar{g}_{m_0}\|_{C^k(\bar{g}_{m_0})}+\|\bar{k}-\bar{k}_{m_0}\|_{C^k(\bar{g}_{m_0})}\le C\delta^{\frac{1}{2}}$$
if $\delta$ is sufficiently small, where $(\bar{g}_{m_0},\bar{k}_{m_0})$ is the initial data of the Schwarzschild metric $g_{m_0}$ induced on $\Sigma_{II}^+$, and  $C$ is a constant independent of $\delta$.

In section \ref{gluing}, we will prove the following theorem:
\begin{theorem}\label{thm:gluing}
Given a large integer $k$. If $\delta$ is sufficiently small, there exists an initial data set $(\widetilde{g},\widetilde{k})$ on $H$ such that $(\widetilde{g},\widetilde{k})$ coincides with $(\bar{g},\bar{k})$ near the inner boundary of $H$, i.e., $H\cap \Cb_\delta$, and is Kerrian near the outer boundary of $H$, i.e., $H\cap \Cb_{\delta+\varepsilon_0}$. Moreover, 
$$\|\widetilde{g}-\bar{g}_{m_0}\|_{C^k(\bar{g}_{m_0})}+\|\widetilde{k}-\bar{k}_{m_0}\|_{C^k(\bar{g}_{m_0})}\le C\delta^{\frac{1}{2}}$$
for some constant $C$ independent of $\delta$.
\end{theorem}

%It is easy then to extend $\Sigma_G$ to obtain an asymptotically flat complete Cauchy data. To extend inside, we choose two smooth functions $u=\Sigma_M(\ub)$ for $\ub\le0$ and $u=\Sigma_C(\ub)$ for $0\le\ub\le\delta$. We require that $\Sigma_C$ and $\Sigma_G$ can be patched smoothly at $\ub=\delta$ and $\Sigma_C$ and $\Sigma_M$ can be patched smoothly at $\ub=0$. We require that $\Sigma_M$ and $\Sigma_C$ are strictly decreasing so that the hypersurfaces $\ub=\Sigma_M(u)$ and $\ub=\Sigma_C(u)$ are spacelike. We also choose $\Sigma_M$ with constant derivative $-1$ so that the hypersurface $u=\Sigma_M(\ub)$ is smooth at the central line $\ub=u$.  Finally, we use again $\Sigma_M$ and $\Sigma_C$ to denote the hypersurfaces $\ub=\Sigma_M(u)$, $\ub=\Sigma_C(u)$ respectively. 

After deforming the initial data, we use again $\Sigma_{II}^+$ to denote this hypersurface $H$, that is, $\Sigma_{II}^+$ and $H$ are the same differential manifold, but equipped with different initial data. It is then easy to extend $\Sigma_{II}^+$ inside to obtain $\Sigma_{III}^+$ and $\Sigma_{IV}^+$, which are chosen to be in the form $u=f(\ub)$ for some decreasing function $f$. % We choose two smooth functions $u=\Sigma_{IV}^+(\ub)$ for $\ub\le0$ and $u=\Sigma_{III}^+(\ub)$ for $0\le\ub\le\delta$. We require that $\Sigma_{III}^+$ and $\Sigma_{II}^+$ can be patched smoothly at $\ub=\delta$ and $\Sigma_{III}^+$ and $\Sigma_{IV}^+$ can be patched smoothly at $\ub=0$. We require that $\Sigma_{III}^+$ and $\Sigma_{IV}^+$ are strictly decreasing so that the hypersurfaces $u=\Sigma_{III}^+(\ub)$ and $u=\Sigma_{IV}^+(\ub)$ are spacelike. We also choose $\Sigma_{IV}^+$ in the way that the hypersurface $u=\Sigma_{IV}^+(\ub)$ is smooth at the central line $\ub=u$.  Finally, we use again $\Sigma_{III}^+$ and $\Sigma_{IV}^+$ to denote the hypersurfaces $u=\Sigma_{III}^+(\ub)$, $u=\Sigma_{IV}^+(\ub)$ respectively.   
It is also easy to simply extend the initial data all the way up to the spatial infinity of the Kerr spacetime, to obtain $\Sigma_{I}^+$. Moreover, $\Sigma_{III}^+$, $\Sigma_{IV}^+$ and $\Sigma_{I}^+$ are chosen such that $\Sigma^+=\Sigma_{I}^+\cup\Sigma_{II}^+\cup\Sigma_{III}^+\cup\Sigma_{IV}^+$ is a smooth, complete asymptotically flat Cauchy data.%One interesting point is, notice that near the outer boundary of $\Sigma_G$, there should be a piece of $\Sigma_K$ which is also inside the Kerr black hole and $\Sigma_K$ should cross the event horizon. We denote the whole complete initial data by $\Sigma=\Sigma_M\bigcup\Sigma_C\bigcup\Sigma_G\bigcup\Sigma_K$. 

Let us consider the future development of $\Sigma^+$. It is clear that the maximal future development of $\Sigma_I^+$ is isometric to a Kerr spacetime, simply because $\Sigma_I^+$ is a Kerrian initial data. Since part of $\Sigma_I^+$ is included inside the black hole, the maximal future development of $\Sigma_I^+$ contains a full neighborhood of the whole future null infinity of the Kerr spacetime, up to the future timelike infinity. In particular, the maximal future development $\Sigma^+$ possesses a complete future null infinity. %In addition, the future development of $\Sigma_K$ contains a part which is inside the black hole. Notice that our construction does not rule out the possibility that $\Sigma_K$ is exactly Schwarzschildean, in this case, the future boundary of the maximal future development of $\Sigma_K$ contains a singular part which is exactly Schwarzschildean, and if $\Sigma_K$ is not exactly Schwarzschildean, then the future boundary of the maximal future development contains a part of smooth Cauchy horizon.

\subsection{Step 3: Construction of the spacetime outside the black hole} In this subsection, we carry out Step 3 in the construction.  In this step, we need to solve the vacuum Einstein equations starting from $\Sigma^+$ \emph{to the past}, up to some another Cauchy surface $\Sigma$ that does not intersect the black hole region. Let us denote this past solution by $g^*$. The regions $I$, $II$, $III$ and $IV$ in Figure \ref{fig:limei2} (also see Figure \ref{fig:limei3}) are solved sequentially as follows.

% In the above subsection, we have constructed a complete asymptotically flat initial data set $\Sigma$, the maximal future development of which contains a part we know what exactly it is. However, by our construction, there should be a closed trapped surface on $\Sigma$ so that $\Sigma$ is not suitable for playing the role of an actural initial state. We should ask can $\Sigma$ be evolved from another asymptotically Cauchy data without any closed trapped surfaces? At this stage, we only know $\Sigma_M\bigcup\Sigma_C\bigcup H$ (before deforming) is evolved from a characteristic initial data without closed trapped surfaces (even from the past null infinity). Therefore, we need to solve the vacuum Einstein equations from $\Sigma$ \textit{to the past}. We hope the backward solution will exist to some fixed past time, such that we can find some Cauchy hypersurface without closed trapped surfaces. 

\begin{figure}[htbp]
\centering
\includegraphics [width=4 in]{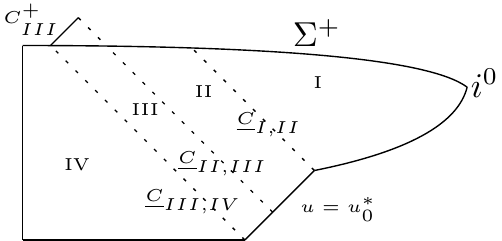}  
\caption{}
 \label{fig:limei3} 
\end{figure}

\subsubsection{Region $I$}\label{regionI} Because $\Sigma_I^+$ is exactly a Cauchy hypersurface in a Kerr spacetime, then the solution $g^*$ in its past domain of dependence must be isometric to the Kerr spacetime. Denote $\Cb_{I,II}$ be the part of the boundary of the past of $\Sigma_I^+$ Moreover, by Theorem \ref{thm:gluing}, the angular momentum per unit mass $a$ satisfies 
$$|a|\le C\delta^{\frac{1}{2}}$$
for some constant $C$ independent of $\delta$. This implies that, given any large number, say $A>0$, if $\delta$ is sufficiently small, those past null geodesic generators of the boundary of the past of $\Sigma_I^+$, have affine lengths larger than $A$, measuring from $\Sigma^+_I\cap\Sigma^+_{II}$,  before any conjugate or cut points. Let us denote $\Cb_{I,II}$ be the truncated null cone consisting of these past null geodesic generators with sufficiently large affine lengths (which will be  determined in the next subsection) before any conjugate or cut points. Moreover,  the data induced on $\Cb_{I,II}$, is $\delta^{\frac{1}{2}}$-close to a null cone in a Schwarzschild spacetime with mass $m_0$ in a suitable topology.

\subsubsection{Region $II$} 
Recall that we have fixed the background Schwarzschild metric $g_{m_0}$ already defined in the region $\delta\le\ub\le\delta+\varepsilon_0$, $u_0\le u\le u_1$ in the form \eqref{gm0}. This Schwarzschild metric $g_{m_0}$ can be extended smoothly (in fact analytically) to the region $\delta\le\ub\le\delta+\varepsilon_0$, $-\infty< u\le u_1$ where the functions $\ub,u$, being optical relative to $g_{m_0}$, are extended smoothly to the past null infinity. The truncated null cone $\Cb_{I,II}$ is then identified to the null cone $\Cb_{\delta+\varepsilon_0}$. By Theorem \ref{thm:gluing}, the data induced on $\Sigma_{II}^+$ is $\delta^{\frac{1}{2}}$-close to the data induced by the background Schwarzschild metric $g_{m_0}$. For the data induced on the truncated null cone $\Cb_{I,II}$, given any $u_0^*<0$ with large absolute value, we can choose the number $A$ in the above subsection \ref{regionI}, and $\delta$ sufficiently small, such that the past null generators of $\Cb_{I,II}$ have no conjugate or cut points before $u=u_0^*$, and the data is $\delta^{\frac{1}{2}}$-close to the data induced by $g_{m_0}$ on the corresponding part of $\Cb_{\delta+\varepsilon_0}$. By Cauchy stability, the solution $g^*$ in the past domain of dependence of $\Sigma_{II}^+\cup\Cb_{I,II}$ exists in the region $\delta\le\ub\le\delta+\varepsilon_0$, $u_0^*\le u\le\Sigma_{II}^+(\ub)$, the region $II$. The functions $\ub,u$ are still optical relative to $g^*$. Moreover, $g^*$ and $g_{m_0}$ are $\delta^{\frac{1}{2}}$-close to each other.

\subsubsection{Region $III$} From the construction of region $II$, we know that the data induced on the null cone $\Cb_{II,III}$ is $\delta^{\frac{1}{2}}$-close to the data induced by the background Schwarzschild metric $g_{m_0}$. And the data induced on $\Sigma_{III}^+$ satisfies the estimates established in Step 1. We will solve the vacuum Einstein equations with these two parts of initial data sets. It would be better to take the data induced on the future boundary of the future of $\Sigma_{III}^+$ as part of the initial data instead of the data induced on $\Sigma_{III}^+$ itself because the problem can then be formulated in terms of a characteristic initial value problem. The future boundary of the future of $\Sigma_{III}^+$ consists of a null cone, denoted by $C_{III}^+$, emanating from $\Sigma_{III}^+\cap\Sigma_{IV}^+$ and the part of $\Cb_{\delta}$ in the future of $\Sigma_{III}^+$. Assume that $u=u_1^*$ on $C_{III}^+$. Then region $III$ can be solved to the past starting from the characteristic initial data on $\Cb_{\delta}$ and $C_{u_1^*}$. The data induced on $\Cb_{\delta}$ is $\delta^{\frac{1}{2}}$-close to the data induced by $g_{m_0}$, and the data induced on $C_{u_1^*}$ still satisfies the estimates established in Step 1 because the data induced on $\Sigma_{III}^+$ does not change after the local deformation in Step 2.

The existence of the solution is established in Section \ref{aprioriestimates}, in which Theorem \ref{step1} will also be proved. The existence result can be summarized as
\begin{theorem}\label{step3}
Given  $u_0^*<0$ with large absolute value. With the initial data on $\Cb_{\delta}$ and $C_{u_1^*}$, if $\delta$ is sufficiently small, the past solution of the vacuum Einstein equation exists for $0\le\ub\le\delta, u_0^*\le u\le u_1^*$.
\end{theorem}

\subsubsection{Region $IV$} The last step is to solve region $IV$, considered to be the past development of the initial data induced on $\Sigma_{IV}^+$ and $\Cb_{III,IV}$, the part of $\Cb_0$ where $u_0^*\le u\le u_1^*$. The data induced on $\Sigma_{IV}^+$ is in fact Minkowskian and the data induced on $\Cb_{III,IV}$ satisfies the estimates established in solving region $III$. The data induced on $\Cb_{III,IV}$ is a priori large so the existence of the solution in the whole region $IV$ is not quite clear. To overcome this difficulty, a key observation is that the data on $\Cb_{III,IV}$ can in addition be proved to be close to being Minkowskian. To state this precisely, let us extend the optical functions $\ub,u$ to region where $u_0^*\le u\le \ub\le 0$, and define the background Minkowski metric $g_0$ in the way that it is written in this region as 
$$g_0=-4\D\ub\D u+r^2\D\sigma_{\mathbb{S}^2}$$
with $r=\ub-u$. Then we have the following theorem.
\begin{theorem}\label{thm:closetoM}
If $\delta$ is sufficiently small, then the data induced on $\Cb_{III,IV}$ is $\delta^{\frac{1}{2}}$-close to the data induced from the Minkowski metric $g_0$ in a sufficient regular sense.
\end{theorem}

This theorem will also be proved in Section \ref{aprioriestimates}. Then by Cauchy stability, if $\delta$ is sufficiently small, the solution can be solved in  the region $IV$, where $\ub\le0$, $\ub+u\ge u_0^*$. The solution is still $\delta^{\frac{1}{2}}$-close to the Minkowski metric $g_0$. Note that $\ub$ and $u$ are not necessarily optical functions relative to the solution $g^*$.

\subsubsection{The Cauchy hypersurface $\Sigma$} At last, we only need to pick a Cauchy hypersurface $\Sigma$ in the solution $g^*$ such that no points in $\Sigma$ is contained in the black hole region. Let us also denote by $g^*$ the maximal (future and past) development of $\Sigma^+$.  It is clear that the maximal future development of $\Sigma^+$ contains a complete future null infinity and a closed trapped surface. Therefore $g^*$ contains a black hole region $\mathcal{B}$. From the above subsections, the maximal past development of $\Sigma^+$ contains  regions $I$, $II$, $III$ and $IV$. The size of these regions is characterized by the number $u^*_0$: For any $u^*_0<0$ with large absolute value, the regions $I$, $II$, $III$ and $IV$ can be solved if $\delta$ is chosen sufficiently small. Then we can pick another Cauchy hypersurface $\Sigma$ in the following way (see Figure \ref{fig:limei4}):
\begin{itemize}[leftmargin=0.5cm]
%\item The part of $\Sigma$ in region $I$, denoted by $\Sigma_{I}$, is chosen to be some constant $t$ slice in the Kerr spacetime in the Boyer–Lindquist coordinate, such that the intersection of $\Sigma_I$ and $\Cb_{I,II}$ is the spherical section $S_{\delta+\varepsilon_0,u_0^*}$ where $\ub=\delta+\varepsilon_0$, $u=u_0^*$.
\item The part of $\Sigma$ in region $IV$ is chosen to be a spacelike hypersurface expressed in the equation $\ub+u=u_0^*+1$, where $\frac{1}{2}(u_0^*+1)\le\ub\le0$. It intersects the timelike curve $\ub=u$ at $\ub=u=\frac{1}{2}(u_0^*+1)$ and $\Cb_{III,IV}$, i.e., $\ub=0$ at the sphere $\ub=0, u=u_0^*+1$, which is to the future of the null cone $u=u_0^*$. Because $g^*$ is $\delta^{\frac{1}{2}}$-close to $g_0$ and $\Sigma_{IV}$ is spacelike relative to $g_0$, then if $\delta$ is sufficiently small, $\Sigma_{IV}$ is spacelike relative to $g^*$.
\item The part of $\Sigma$ in regions $II$ and $III$, is chosen to be expressed in an equation of the form $u=h(\ub)$, $0\le\ub\le\delta+\varepsilon_0$. $h$ should satisfy the following properties: (i) $h'<0$, which is equivalent to  that $\Sigma$ is spacelike; (ii) $h^{(n)}(0)=-1$ for all positive integer $n$, which is equivalent to that $\Sigma$ is smooth across the boundary of the previous part.
\item The part of $\Sigma$ in region $I$, the region isometric to the Kerr spacetime, is chosen in the way that it is spacelike all the way up to the spatial infinity, and smoothly connected to the previous part of $\Sigma$.
\end{itemize}
Obviously the hypersurface $\Sigma$ constructed above is a Cauchy hypersurface of $g^*$, which contains a black hole region $\mathcal{B}$.  The last thing is to prove 
\begin{proposition}$\Sigma$ contains no points in the black hole region, i.e., $\Sigma\cap \mathcal{B}=\varnothing$, if  $u_0^*$ has sufficiently large absolute value.
\end{proposition}
\begin{figure}[htbp]
\centering
\includegraphics [width=4 in]{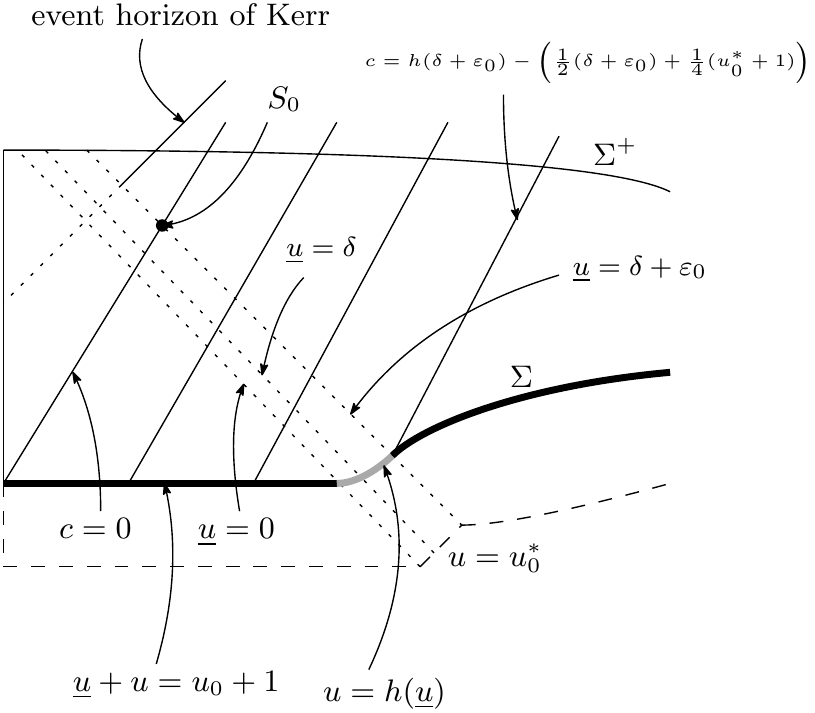}  
\caption{}
 \label{fig:limei4} 
\end{figure}
\begin{proof}
 Consider the $c$-parametrized family of hypersurfaces $$u=\frac{1}{2}\ub+\frac{1}{4}(u_0^*+1)+c$$ for $\frac{1}{2}(u_0^*+1)\le \ub\le\delta+\varepsilon_0$. For any 
 $$c\in\left[h(\delta+\varepsilon_0)-\left(\frac{1}{2}(\delta+\varepsilon_0)+\frac{1}{4}(u_0^*+1)\right),0\right],$$
  the hypersurface intersects $\Sigma$ in regions $II$, $III$ or $IV$.  The hypersurface $c=h(\delta+\varepsilon_0)-\left(\frac{1}{2}(\delta+\varepsilon_0)+\frac{1}{4}(u_0^*+1)\right)$ intersects $\Sigma$ at the sphere $\ub=\delta+\varepsilon_0, u=h(\delta+\varepsilon_0)$, and the hypersurface $c=0$  is a cone with vertex $\ub=u=\frac{1}{2}(u_0^*+1)$ on $\Sigma$.  These hypersurfaces are timelike in region $IV$ relative to $g_0$, so if $\delta$ is sufficiently small, they are also timelike relative to $g^*$. They are also timelike in regions $II$ and $III$ where $\ub,u$ are optical functions. 
  
  The intersection of the timelike hypersurface with parameter $c$ and $\Cb_{I,II}$, i.e., $\ub=\delta+\varepsilon_0$, is the sphere 
  $\ub=\delta+\varepsilon_0, u=\frac{1}{2}(\delta+\varepsilon_0)+\frac{1}{4}(u_0^*+1)+c.$ In particular, the intersection of the hypersurface $c=0$ and $\Cb_{I,II}$ is the sphere $S_0:\ub=\delta+\varepsilon_0, u=\frac{1}{2}(\delta+\varepsilon_0)+\frac{1}{4}(u_0^*+1)$.  It is clear that  if $u_0^*$ has sufficiently large absolute value (note that $m_0$ is however fixed), the radial function $r_{m_0}$ of the background Schwarzschild metric $g_{m_0}$ is larger than $2m_0+1$ at $S_0$ and hence $S_0$ belongs to the domain of outer communication of $g_{m_0}$. Then by the $\delta^{\frac{1}{2}}$-closeness of the data induced on $\Cb_{I,II}$ of the constructed Kerr metric to the background Schwarzschild metric $g_{m_0}$, $S_0$ still belongs to the domain of outer communication of the exterior Kerr spacetime\footnote{In fact, the Kretschmann scalar $\mathbf{R}_{\alpha\beta\mu\nu}\mathbf{R}^{\alpha\beta\mu\nu}$ relative to $g_{m_0}$ at $S_0$ is smaller than $\frac{48m_0^2}{(2m_0+1)^6}$, then if $\delta$ is sufficiently small, the Kretschmann scalar relative to the constructed Kerr metric at $S_0$ is still smaller than $\frac{48m_0^2}{(2m_0+1)^6}$. On the other hand, if $\delta$ is sufficiently small, the Kretschmann scalar relative to the constructed Kerr metric at the points inside the Kerr black hole and  before $\Sigma^+$ should be larger than $\frac{48m_0^2}{(2m_0)^6}-O(\delta^{\frac{1}{2}})$.}. Therefore, the intersections of the $c$-parametrized timelike hypersurfaces and $\Cb_{I,II}$, which are at the past of $S_0$, also belong to the domain of outer communication of the exterior Kerr spacetime.
  
  So, for every point on the part of $\Sigma$ in regions $II$, $III$ and $IV$, there exists a timelike curve connecting this point and some other point on the sphere $\ub=\delta+\varepsilon_0, u=\frac{1}{2}(\delta+\varepsilon_0)+\frac{1}{4}(u_0^*+1)+c$ for some $c$. On the other hand, every point on the part of $\Sigma$ in region $I$ belongs to the domain of outer communication of the Kerr spacetime. Therefore, for every point on $\Sigma$, there exists a timelike curve starting from this point and extending to the future null infinity. This implies that $\Sigma\cap\mathcal{B}=\varnothing$.
 
 \end{proof}

The proof of the Main Theorem is then completed. In the following sections, we will give the proofs of Theorems \ref{step1}, \ref{thm:gluing}, \ref{step3} and \ref{thm:closetoM}.

\section{Proofs of Theorems \ref{step1}, \ref{step3} and \ref{thm:closetoM}}\label{aprioriestimates}

In this section, we study a characteristic initial value problem with initial data given on two intersecting null cone emanating from a sphere. We will prove the existence of the solution of the vacuum Einstein equations by proving appropriate  a priori estimates. Moreover, we will derive some improved estimates based on condition \eqref{integrate=m0}. Then Theorems \ref{step1}, \ref{step3} and \ref{thm:closetoM} will follow immediately. 

\subsection{The characteristic initial value problem and the existence theorem}
Recall that in a spacetime $(M,g)$ we denote $u$ and $\ub$ be two optical functions on $M$, that is $$g(\nabla u,\nabla u)=g(\nabla\ub,\nabla\ub)=0.$$
 $M$ is then foliated by the level sets of $\ub$ and $u$ respectively, and the functions $u$ and $\ub$ increase towards the future. We use $C_u$ to denote the outgoing null hypersurfaces which are the level sets of $u$ and use ${\Cb}_{\ub}$ to denote the incoming null hypersurfaces which are the level sets of $\ub$. We denote the intersection $S_{\ub,u}=\Cb_{\ub} \cap C_u$, which is a  spacelike two-sphere for each pair of $\ub, u$.

The lapse function $\Omega$ is defined by the formula
$$ \Omega^{-2}=-2g(\nabla\ub,\nabla u).$$
  We then define the normalized null pair $e_3,e_4$ with $g(e_3,e_4)=-1$ by
  $$e_3=\Lbh=-2\Omega\nabla\ub,\ e_4=\Lh=-2\Omega\nabla u,$$ and define one another null pair
  $$\Lb=\Omega \Lbh,\ L=\Omega \Lh,$$
preserving the double null foliation. We also choose a local frame $
   {e_1,e_2}$ tangent to $S_{\ub,u}$. We call $\{e_1, e_2, e_3,e_4\}$ a null frame.  Under this null frame,  we have a null decomposition of the connection coefficients as follows:
\begin{align*}
\chi_{AB}&=g(\nabla_A\Lh,e_B),\quad \eta_A=-\frac{1}{2}g(\nabla_{\Lbh}e_A,\Lh),\quad \omega=\frac{1}{2}\Omega g(\nabla_{\Lh}\Lbh,\Lh),\\
\chib_{AB}&=g(\nabla_A\Lbh,e_B), \quad\etab_A=-\frac{1}{2}g(\nabla_{\Lh}e_A,\Lbh), \quad\omegab=\frac{1}{2}\Omega g(\nabla_{\Lbh}\Lh,\Lbh).
\end{align*}
These are tangential tensorfields by which we mean the tensorfield whose contractions with $e_3$ or $e_4$ are zero. We will also use $\zeta=\frac{1}{2}(\eta-\etab)$. The trace of $\chi$ and $\chib$ are denoted by
 $$\tr\chi = \gs^{AB}\chi_{AB},\ \tr\chib = \gs^{AB}\chib_{AB}.$$
 We may also define the trace-free part of $\chi$ and $\chib$ as
 $$\chih=\chi-\frac{1}{2}\tr\chi\gs,\ \chibh=\chib-\frac{1}{2}\tr\chib\gs.$$
  By definition, we can check directly the following useful identities:
  $$\ds\log\Omega=\frac{1}{2}(\eta+\etab),\ D\log\Omega=\omega,\ \Db\log\Omega=\omegab$$
  where $\nablas$ is the covariant derivative of $\gs$ which is the metric induced on $S_{\ub,u}$, and $D$, $\Db$ are the restrictions of Lie derivatives relative to $L$, $\Lb$ on $S_{\ub,u}$. Acting on functions, $D$, $\Db$ are simply the ordinary derivatives relative to $L$, $\Lb$.

We can also define the null components of the curvature tensor
{\bf R}:
\begin{align*}
\alpha_{AB}=\mathbf{R}(e_A,\Lh,e_B,\Lh),&\quad\alphab_{AB}=\mathbf{R}(e_A,\Lbh,e_B,\Lbh),\\
\beta_A=\frac{1}{2}\mathbf{R}(e_A,\Lh,\Lbh,\Lh),&\quad\betab_A=\frac{1}{2}\mathbf{R}(e_A,\Lbh,\Lbh,\Lh),\\
\rho=\frac{1}{4}\mathbf{R}(\Lbh,\Lh,\Lbh,\Lh),&\quad\sigma=\frac{1}{4}\mathbf{R}(\Lbh,\Lh,e_A,e_B)\epsilons^{AB}
\end{align*}
where $\epsilons$ is the volume form of the induced metric $\gs$ on $S_{\ub,u}$. These are all tangential tensorfields.

Now let us fix a small parameter $\delta>0$ and two numbers $u_0<u_1<0$. The characteristic initial data are given on $C_{u_0}\cup\Cb_0$. Here $C_{u_0}$ refers to its part where $0\le\ub\le\delta$ and $\Cb_0$ refers to its part where $u_0\le u\le u_1$. We expect that the solution can be solved in the region $0\le \ub\le\delta$, $u_0\le u\le u_1$. We introduce some norms in this spacetime region. These norms are designed according to the new hierarchy \eqref{hierarchy}. First of all, for an arbitrary tensorfield $\phi$, we denote 
\begin{align*}\|\phi\|_{H^k(\ub,u)}&=\sum_{i=0}^k\left(\int_{S_{\ub,u}}|\nablas^i\phi|^2\D\mu_{\gs}\right)^{\frac{1}{2}}\\
\|\phi\|_{L^p_{\ub}H^k(u)}&=\left(\int_0^{\delta}\|\phi\|^p_{H^k(\ub,u)}\D\ub\right)^{\frac{1}{p}}\\
\|\phi\|_{L^p_{u}H^k(\ub)}&=\left(\int_{u_0}^{u_1}\|\phi\|^p_{H^k(\ub,u)}\D u\right)^{\frac{1}{p}}
\end{align*}
for $2\le p\le \infty$. For curvature components, we denote
\begin{gather*}
\mathcal{R}_k[\alpha](u)=\delta\|\alpha\|_{L^2_{\ub}H^k(u)}\\
\mathcal{R}_k[\beta](u)=\|\beta\|_{L^2_{\ub}H^k(u)}\\
\mathcal{R}_k[\rho](u)=\delta^{-\frac{1}{2}}\|\rho\|_{L^2_{\ub}H^k(u)}\\
\mathcal{R}_k[\sigma](u)=\delta^{-\frac{1}{2}}\|\sigma\|_{L^2_{\ub}H^k(u)}\\
\mathcal{R}_k[\betab](u)=\delta^{-\frac{1}{2}}\|\betab\|_{L^2_{\ub}H^k(u)}
\end{gather*}
and denote
\begin{gather*}
\underline{\mathcal{R}}_k[\beta](\ub)=\delta\|\beta\|_{L^2_{u}H^k(\ub)}\\
\underline{\mathcal{R}}_k[\rho](\ub)=\|\rho\|_{L^2_{u}H^k(\ub)}\\
\underline{\mathcal{R}}_k[\sigma](\ub)=\|\sigma\|_{L^2_{u}H^k(\ub)}\\
\underline{\mathcal{R}}_k[\betab](\ub)=\delta^{-\frac{1}{2}}\|\betab\|_{L^2_{u}H^k(\ub)}\\
\underline{\mathcal{R}}_k[\alphab](\ub)=\delta^{-\frac{1}{2}}\|\alphab\|_{L^2_{u}H^k(\ub)}.
\end{gather*}
Let $R$ be one of the curvature components $\alpha$, $\beta$, $\rho$, $\sigma$, $\betab$ and $\underline{R}$ to be one of the curvature components $\beta$, $\rho$, $\sigma$, $\betab$, $\alphab$, we also denote
\begin{align*}
\mathcal{R}_k[R]=\sup_{u_0\le u\le u_1}\mathcal{R}_k[R](u), \ \underline{\mathcal{R}}_k[\underline{R}]=\sup_{0\le \ub\le \delta}\underline{\mathcal{R}}_k[\underline{R}](\ub).
\end{align*}
We will also denote 
$$\mathcal{R}_k=\sum_{R\in\{\alpha,\beta,\rho,\sigma,\betab\}}\mathcal{R}_k[R],\ \underline{\mathcal{R}}_k=\sum_{\underline{R}\in\{\beta,\rho,\sigma,\betab,\alphab\}}\underline{\mathcal{R}}_k[\underline{R}]$$ 
and the initial norm on $C_{u_0}\cup\Cb_0$
$$\mathcal{R}_k^{(0)}=\sum_{R\in\{\alpha,\beta,\rho,\sigma,\betab\}}\mathcal{R}_k[R](u_0)+\sum_{\underline{R}\in\{\beta,\rho,\sigma,\betab,\alphab\}}\underline{\mathcal{R}}_k[\underline{R}](0).$$
Also, we denote
\begin{gather*}
\mathcal{O}_k[\chih](\ub,u)=\delta^{\frac{1}{2}}\|\chih\|_{H^k(\ub,u)}\\
\mathcal{O}_k[\tr\chi](\ub,u)=\left\|\tr\chi\right\|_{H^k(\ub,u)}\\
\mathcal{O}_k[\chibh](\ub,u)=\delta^{-\frac{1}{2}}\|\chibh\|_{H^k(\ub,u)}\\
\mathcal{O}_k[\widetilde{\tr\chib}](\ub,u)=\delta^{-\frac{1}{2}}\left\|\tr\chib+\frac{2}{|u|}\right\|_{H^k(\ub,u)}\\
\mathcal{O}_k[\eta](\ub,u)=\delta^{-\frac{1}{2}}\|\eta\|_{H^k(\ub,u)}\\
\mathcal{O}_k[\etab](\ub,u)=\delta^{-\frac{1}{2}}\|\etab\|_{H^k(\ub,u)}\\
\mathcal{O}_k[\omega](\ub,u)=\|\omega\|_{H^k(\ub,u)}\\
\mathcal{O}_k[\omegab](\ub,u)=\delta^{-\frac{1}{2}}\|\omegab\|_{H^k(\ub,u)}\\
\end{gather*}
Let $\Gamma$ be one of the connection coefficients $\chih$, $\tr\chi$, $\chibh$, $\widetilde{\tr\chib}=\tr\chib+\frac{2}{|u|}$, $\eta$, $\etab$, $\omega$, $\omegab$. Denote
\begin{align*}
\mathcal{O}_k[\Gamma]=\sup_{u_0\le u\le u_1}\sup_{0\le \ub\le \delta}\mathcal{O}_k[\Gamma](\ub, u).
\end{align*}
We will also denote 
$$\mathcal{O}_k=\sum_{\Gamma\in\{\chih, \tr\chi, \chibh, \widetilde{\tr\chib}, \eta, \etab, \omega, \omegab\}}\mathcal{O}_k[\Gamma]$$
and the initial norm on $C_{u_0}\cup\Cb_0$
\begin{align*}\mathcal{O}_k^{(0)}=\sup_{[\ub,u]\in \{0\}\times[u_0,u_1] \cup [0,\delta]\times\{u_0\}}\left(\sum_{\Gamma\in\{\chih, \tr\chi, \chibh, \widetilde{\tr\chib}, \eta, \etab, \omega, \omegab\}}\mathcal{O}_k[\Gamma](\ub,u)\right.\\
\left.+\delta^{-\frac{1}{2}}\|\log\Omega\|_{H^k(\ub,u)}+\delta^{-\frac{1}{2}}\left||u|^{-2}\gs_{AB}-\overset{\circ}{\gs}_{AB}\right|\right)\end{align*}
where $\overset{\circ}{\gs}_{AB}$ is the standard unit round metric on $\mathbb{S}^2$. Then we are going to prove the following existence theorem.
\begin{theorem}\label{existence}
Suppose that the initial data given on $C_{u_0}\cup\Cb_0$ obey the estimates
\begin{align*}
\mathcal{O}_k^{(0)},\mathcal{R}_k^{(0)}<\infty
\end{align*}
for some $k\ge2$. Then if $\delta$ is sufficiently small (depending on $\mathcal{O}_k^{(0)},\mathcal{R}_k^{(0)}$ and $u_0,u_1$), the smooth solution of the vacuum Einstein equations remains regular in $0\le\ub\le\delta$, $u_0\le u\le u_1$. Moreover, the solution obeys the following estimates
\begin{align*}
\mathcal{O}_k,\mathcal{R}_k,\underline{\mathcal{R}}_k\lesssim C(\mathcal{O}_k^{(0)},\mathcal{R}_k^{(0)})
\end{align*}
where $C$ is a constant depending on $\mathcal{O}_k^{(0)},\mathcal{R}_k^{(0)}$.
\end{theorem}
{\bf Here and in the following proof, $A\lesssim B$ means $A\le cB$ for some universal constant $c$, which may also depend on $u_0$ and $u_1$.}
\begin{proof}
We will only prove the a priori estimates. The construction of the solution is then routine (see \cite{Chr08}) and we will not present here. The a priori estimates are proved by first assuming the solution of the vacuum Einstein equations exists in the region $0\le\ub\le\delta$, $u_0\le u\le u_1$. We will show that if $\delta$ is sufficiently small,
\begin{align*}
\mathcal{O}_k,\mathcal{R}_k,\underline{\mathcal{R}}_k\lesssim C(\mathcal{O}_k^{(0)},\mathcal{R}_k^{(0)}).
\end{align*}
This can be achieved by proving two propositions.
\begin{proposition}\label{curvatureconnection}
Assume that $\mathcal{R}_k,\underline{\mathcal{R}}_k<\infty$ for $k\ge2$, then  if $\delta$ is sufficiently small depending on $\mathcal{O}^{(0)}_k$ and $\mathcal{R}_k$, $\underline{\mathcal{R}}_k$, then
$$\mathcal{O}_k\lesssim C(\mathcal{O}^{(0)}_k,\mathcal{R}_k,\underline{\mathcal{R}}_k).$$
In particular, 
$$\mathcal{O}_k[\chih]\lesssim\mathcal{O}^{(0)}_k+\mathcal{R}_k[\alpha]+1,\ \mathcal{O}_k[\chibh]\lesssim\mathcal{O}^{(0)}_k+\underline{\mathcal{R}}_k[\alphab]+1.$$

\end{proposition}

\begin{proof}
The proof is began by a bootstrap assumption
\begin{align}\label{bootstrap}
\mathcal{O}_k\le \Delta.
\end{align}
We should first obtain certain geometric inequalities. Let $I(\ub,u)$ is the isoperimetric constant of $S_{\ub,u}$. We introduce the geometric bootstrap assumptions
\begin{align}\label{geometricbootstrap}
\frac{1}{4}\le \Omega\le 4,\ I(\ub,u)\le\frac{2}{\pi},\ \pi|u|^2\le Area(\ub,u)\le 16\pi|u|^2.
\end{align}
Recall the following form of the Sobolev inequalities for any tangential tensorfield $\phi$ in Section 5.2 of \cite{Chr08}:
\begin{align*}
&\|\phi\|_{L^4(\ub,u)}\lesssim \sqrt{\max\{I(S_{\underline{u},u}),1\}}\sum_{i=0}^{1}r^{-\frac{1}{2}}\|(r\nablas)^i\phi\|_{L^2(\ub,u)}\\
&\|\phi\|_{L^\infty(\ub,u)}\lesssim \sqrt{\max\{I(S_{\underline{u},u}),1\}}\sum_{i=0}^{2}r^{-1}\|(r\nablas)^i\phi\|_{L^2(\ub,u)}
\end{align*}
where $r=r(\underline{u},u)$ is the area radius defined by $4\pi r^2=Area(\ub,u)$, and that $|u_1|\le|u|\le|u_0|$, we immediately have
\begin{equation}\label{Sobolev}
\begin{split}
\|\phi\|_{L^4(\ub,u)}\lesssim&\|\phi\|_{H^1(\ub,u)},\\
\|\phi\|_{L^\infty(\ub,u)}\lesssim&\|\phi\|_{H^2(\ub,u)}.
\end{split}
\end{equation}
Then from the bootstrap assumptions for $\chih$, $\tr\chi$, $\chibh$, $\tr\chib$, we have
\begin{align*}
\delta^{\frac{1}{2}}|\chih|, |\tr\chi|, \delta^{-\frac{1}{2}}|\chibh|, \delta^{-\frac{1}{2}}|\widetilde{\tr\chib}|\lesssim \Delta
\end{align*}
for $0\le\ub\le\delta$, $u_0\le u\le u_1$. We are going to improve the bootstrap assumption \eqref{geometricbootstrap}. First of all, to improve the bootstrap assumption on $\Omega$, we simply use the equation $D\Omega=\Omega\omega$ and the Sobolev inequality to obtain
$$\left|\Omega-\Omega\big|_{\ub=0}\right|\le4\sup_u\int_0^\delta\|\omega\|_{L^\infty(\ub,u)}\D\ub\le 4\delta\Delta\le\frac{1}{4}$$
if $\delta\le\frac{1}{16\Delta}$. Note that $\left|\Omega\big|_{\ub=0}-1\right|\lesssim\delta^{\frac{1}{2}}\mathcal{O}_k^{(0)}$ from the definition of $\mathcal{O}^{(0)}_k$, hence if $\delta$ is sufficiently small depending on $\mathcal{O}^{(0)}_k$,
\begin{align}\label{geometricOmega}
\frac{1}{2}\le\Omega\le\frac{3}{2}
\end{align}
and the bootstrap assumption \eqref{geometricbootstrap} on $\Omega$ is improved. To improve the estimate for $Area(\ub,u)$,  we consider the equation $DArea(\ub,u)=\int_{S_{\ub,u}}\Omega\tr\chi\D\mu_{\gs}$ and we have
\begin{align*}
|Area(\ub,u)-Area(0,u)|\le\delta\sup_{0\le\ub\le\delta}(\|\Omega\tr\chi\|_{L^\infty(\ub,u)}Area(\ub,u))\lesssim\delta\Delta\cdot 16\pi|u|^2.
\end{align*}
Since we also have $|Area(0,u)-4\pi|u|^2|\lesssim \delta^{\frac{1}{2}}\mathcal{O}_k^{(0)}$ from the definition of $\mathcal{O}_k^{(0)}$, then if $\delta$ is sufficiently small depending on $\mathcal{O}_k^{(0)}$ and $\Delta$, we will have 
\begin{align}\label{area}
2\pi|u|^2\le Area(S_{\ub,u})\le 8\pi|u|^2
\end{align}
which improves the bootstrap assumption \eqref{geometricbootstrap} on $Area(\ub,u)$. Let $\Lambda(\ub,u)$ and $\lambda(\ub,u)$ to be the larger and smaller eigenvalues of $\gs\big|_{S_{\ub,u}}$ relative to $\gs\big|_{S_{0,u}}$. Note that $\Lambda(0,u)=\lambda(0,u)=1$, then we have the following (see Lemma 5.3 of \cite{Chr08})
\begin{align*}
\sqrt{\Lambda\lambda}(\ub)=\exp\left(\int_0^{\ub}\Omega\tr\chi(\ub')\D\ub'\right),\\
\sqrt{\frac{\Lambda}{\lambda}}(\ub)\le\exp\left(2\int_0^{\ub}|\Omega\chih|(\ub')\D\ub'\right),
\end{align*}
which implies that $|\Lambda-1|,|\lambda-1|\lesssim\delta^{\frac{1}{2}}\Delta$ if $\delta$ is sufficiently small depending on $\Delta$. Then we will have (see Lemma 5.4 of \cite{Chr08})
\begin{align*}
I(\ub,u)\le (1+\delta^{\frac{1}{2}}\Delta)I(0,u).
\end{align*}
Since we also have $\left|I(0,u)-\frac{1}{2\pi}\right|\lesssim \delta^{\frac{1}{2}}\mathcal{O}_k^{(0)}$ from the definition of $\mathcal{O}_k^{(0)}$, then if $\delta$ is sufficiently small depending on $\mathcal{O}_k^{(0)}$ and $\Delta$, we will have 
\begin{align}\label{isoperimetric}
I(\ub,u)\le \frac{1}{\pi}
\end{align}
which improves the bootstrap assumptions \eqref{geometricbootstrap} on $I(\ub,u)$. By a bootstrap argument it follows that the assumptions \eqref{geometricbootstrap} are in fact true. Then \eqref{geometricOmega}, \eqref{area}, \eqref{isoperimetric} and the Sobolev inequality \eqref{Sobolev} are also true. Moreover, suppose that $\phi_1,\cdots,\phi_n$ are $n$ tangential tensor fields and $\phi_1\cdots\phi_n$ represents some kind of contractions of them, then
\begin{equation}\label{Sobolev1}
\begin{split}
\|\phi_1\cdots\phi_n\|_{H^k(\ub,u)}\lesssim&\|\phi_1\|_{H^k(\ub,u)}\cdots\|\phi_n\|_{H^k(\ub,u)},\\
\|\phi_1\cdots\phi_n\|_{H^{k-1}(\ub,u)}\lesssim&\|\phi_1\|_{H^{k-1}(\ub,u)}\|\phi_2\|_{H^{k}(\ub,u)}\cdots\|\phi_n\|_{H^k(\ub,u)},
\end{split}
\end{equation}
for integer $k\ge 2$. These inequalities will be used to handle the nonlinear terms. 
Moreover, we have the estimates for transport equations if $\delta$ is sufficiently small depending on $\Delta$ (see also Section 4.1 of \cite{Chr08})
\begin{equation}\label{transportestimate}
\begin{split}
\|\phi\|_{L^2(\ub,u)}\lesssim&\|\phi\|_{L^2(0,u)}+\int_0^\delta\|D\phi\|_{L^2(\ub,u)}\D \ub,\\
\|\phi\|_{L^2(\ub,u)}\lesssim&\|\phi\|_{L^2(\ub,u_0)}+\int_{u_0}^{u_1}\|\Db\phi\|_{L^2(\ub,u)}\D u.
\end{split}
\end{equation}
Recall that we have the commutation formulas (see also Section 4.1 of \cite{Chr08})
\begin{align*}
[D,\nablas^i]\phi=\sum_{j_1+j_2=i-1}\nablas^{1+j_1}(\Omega\chi)\cdot\nablas^{j_2}\phi, \ [\Db,\nablas^i]\phi=\sum_{j_1+j_2=i-1}\nablas^{1+j_1}(\Omega\chib)\cdot\nablas^{j_2}\phi
\end{align*}
for tangential tensorfields $\phi$ (if $\phi$ is a function, then $j_2\ge1$). Let us introduce an auxiliary bootstrap assumption 
\begin{equation}\label{bootstrapOmegaHk}
\|\Omega\|_{H^k(\ub,u)}\le L
\end{equation}
for some sufficiently large number $L>0$. Then by the Sobolev inequality \eqref{Sobolev1}, for $i\le k$,
\begin{align*}
\|\nablas^i\phi\|_{L^2(\ub,u)}&\lesssim\|\nablas^i\phi\|_{L^2(0,u)}+\int_0^\delta\left(\|\nablas^iD\phi\|_{L^2(\ub,u)}+\|\Omega\|_{H^k(\ub,u)}\|\chi\|_{H^k(\ub,u)}\|\phi\|_{H^i(\ub,u)}\right)\D\ub\\
&\lesssim\|\nablas^i\phi\|_{L^2(0,u)}+\int_0^\delta\|\nablas^iD\phi\|_{L^2(\ub,u)}\D\ub+\delta^{\frac{1}{2}}L\Delta\sup_{0\le\ub\le\delta}\|\phi\|_{H^i(\ub,u)}.
\end{align*}
By taking $\delta$ sufficiently small depending on $\Delta$ and $L$, the last term on the right can be absorbed by the left if we take supremum on $\ub$, we will then have
\begin{align}\label{transportestimate1}
\|\phi\|_{H^k(\ub,u)}\lesssim\|\phi\|_{H^k(0,u)}+\|D\phi\|_{L^1_{\ub}H^k(\ub,u)}.
\end{align}
We are going to improve \eqref{bootstrapOmegaHk}. By the equation $D\log\Omega=\omega$, we have
\begin{align*}
\|\log\Omega\|_{H^k(\ub,u)}&\lesssim\|\log\Omega\|_{H^k(0,u)}+\|\omega\|_{L^1_{\ub}H^k(u)}\\&\lesssim
\|\log\Omega\|_{H^k(0,u)}+\delta\|\omega\|_{L^\infty_{\ub}H^k(u)}\\&\lesssim\delta^{\frac{1}{2}}\mathcal{O}^{(0)}_k+\delta\Delta\lesssim\delta^{\frac{1}{2}}\mathcal{O}^{(0)}_k+\delta^{\frac{1}{2}}
\end{align*}
if $\delta$ is sufficiently small depending on $\Delta$. Here we have used \eqref{transportestimate1} and the inequality
\begin{align}\label{L1Linfty}
\|\cdot\|_{L^1_{\ub}H^k(u)}\lesssim&\delta\|\cdot\|_{L^\infty_{\ub}H^k(u)}.
\end{align}
 The above estimate for $\log\Omega$ also implies that
\begin{align}\label{Omega-1}
\|\Omega-1\|_{H^k(\ub,u)}\lesssim\delta^{\frac{1}{2}}(\mathcal{O}^{(0)}_k+1)
\end{align}
 and in particular 
\begin{align}\label{Omega}
\|\Omega\|_{H^k(\ub,u)}\lesssim1
\end{align}
 if $\delta$ is sufficiently small depending on $\mathcal{O}^{(0)}_k$. We can then choose $L$ large enough such that \eqref{Omega} is an improvement of \eqref{bootstrapOmegaHk}. Then by a bootstrap argument it follows that \eqref{Omega} is true without assuming \eqref{bootstrapOmegaHk} and then \eqref{transportestimate1} and \eqref{Omega-1} are also true, if $\delta$ is sufficiently small depending on $\Delta$ but not on $L$.

Along $\Db$ direction, we have
\begin{align*}
\|\nablas^i\phi\|_{L^2(\ub,u)}&\lesssim\|\nablas^i\phi\|_{L^2(\ub,u_0)}+\int_{u_0}^{u_1}\left(\|\nablas^iD\phi\|_{L^2(\ub,u)}+\|\nablas(\Omega\chib)\|_{H^{k-1}(\ub,u)}\|\phi\|_{H^i(\ub,u)}\right)\D u.
\end{align*}
The term $\|\nablas(\Omega\chib)\|_{H^{k-1}(\ub,u)}$ can be estimated by
\begin{align*}
\|\nablas(\Omega\chib)\|_{H^{k-1}(\ub,u)}\lesssim&\|(\nablas\Omega)\chib\|_{H^{k-1}(\ub,u)}+\|\Omega\nablas\chib\|_{H^{k-1}(\ub,u)}\\
\lesssim&\|\nablas\Omega\|_{H^{k-1}(\ub,u)}\|\chib\|_{H^{k}(\ub,u)}+\|\Omega\|_{H^{k}(\ub,u)}\|\nablas\chib\|_{H^{k-1}(\ub,u)}\\
\lesssim&\delta^{\frac{1}{2}}\Delta(\mathcal{O}_k^{(0)}+1),
\end{align*}
where we have used the second of \eqref{Sobolev1}, \eqref{Omega-1} and \eqref{Omega}. Then we have
\begin{align*}
\|\nablas^i\phi\|_{L^2(\ub,u)}\lesssim\|\nablas^i\phi\|_{L^2(\ub,u_0)}+\int_{u_0}^{u_1}\|\nablas^iD\phi\|_{L^2(\ub,u)}\D u+\delta^{\frac{1}{2}}\Delta(\mathcal{O}_k^{(0)}+1)\sup_{u_0\le u\le u_1}\|\phi\|_{H^i(\ub,u)},
\end{align*}
By taking $\delta$ sufficiently small depending on $\Delta$ and $\mathcal{O}_k^{(0)}$, we have
\begin{align}\label{transportestimate2}
\|\phi\|_{H^k(\ub,u)}\lesssim\|\phi\|_{H^k(\ub,u_0)}+\|\Db\phi\|_{L^1_uH^k(\ub,u)}.
\end{align}
We have then obtained all desired geometric inequalities. 

Before doing estimates for the connection coefficients, we collect the equations we will use in the following. These are part of the null structure equations which can be found  in \cite{Chr08}.
  \begin{align*}
&\widehat{D}\chih=\omega\chih-\Omega\alpha\\
& D\tr\chi=\omega\tr\chi-\frac{1}{2}\Omega(\tr\chi)^2-\Omega|\chih|^2 \\
&\Db\chibh=\omegab\chibh-\Omega\alphab \\
&\Db\tr\chib=\omegab\tr\chib-\frac{1}{2}\Omega(\tr\chib)^2-\Omega|\chibh|^2\\
&D\eta=\Omega(\chi\cdot\underline{\eta}-\beta)\\
&\underline{D}\underline{\eta}=\Omega(\underline{\chi}\cdot{\eta}+\underline{\beta})\\
&D\underline{\omega}=\Omega^2(2(\eta,\underline{\eta})-|\eta|^2-\rho)\\
&\underline{D}\omega=\Omega^2(2(\eta,\underline{\eta})-|\underline{\eta}|^2-\rho)
\end{align*}

The following estimates mainly rely on the inequalities \eqref{Sobolev1} (the first one), \eqref{transportestimate1}, \eqref{Omega} and \eqref{transportestimate2}, and the bootstrap assumption \eqref{bootstrap}. For the first equation of $\Dh\chih$, we have
\begin{align*}
\|\chih\|_{H^k(\ub,u)}&\lesssim\|\chih\|_{H^k(0,u)}+\|\omega\chih\|_{L^1_{\ub}H^k(u)}+\|\Omega\alpha\|_{L^1_{\ub}H^k(u)}\\
&\lesssim\delta^{-\frac{1}{2}}\mathcal{O}^{(0)}_k+\delta\|\omega\|_{L^\infty_{\ub}H^k(u)}\|\chih\|_{L^\infty_{\ub}H^k(u)}+\|\Omega\|_{L^\infty_{\ub}H^k(u)}\cdot\delta^{\frac{1}{2}}\|\alpha\|_{L^2_{\ub}H^k(u)}\\
&\lesssim\delta^{-\frac{1}{2}}\left(\mathcal{O}^{(0)}_k+\mathcal{R}_k[\alpha]\right)+\delta^{\frac{1}{2}}\Delta^2
\end{align*}
where we use the inequalities \eqref{L1Linfty} together with
\begin{equation}\label{L1L2}
\begin{split}
\|\cdot\|_{L^1_{\ub}H^k(u)}\lesssim&\delta^{\frac{1}{2}}\|\cdot\|_{L^2_{\ub}H^k(u)}
\end{split}
\end{equation} 
for the quadratic terms of the connection coefficients and the curvature term respectively. Then if $\delta$ is sufficiently small depending on $\Delta$, we will have
\begin{align*}
\delta^{\frac{1}{2}}\|\chih\|_{H^k(\ub,u)}\lesssim\mathcal{O}^{(0)}_k+\mathcal{R}_k[\alpha]+1,
\end{align*}
that is,
\begin{align}\label{estimatechih}
\mathcal{O}_k[\chih]\lesssim\mathcal{O}^{(0)}_k+\mathcal{R}_k[\alpha]+1.
\end{align}
 
 Now we turn to the equation for $D\tr\chi$. Similar to estimating $\chih$, we will have
 \begin{align*}
\|\tr\chi\|_{H^k(\ub,u)}&\lesssim\|\tr\chi\|_{H^k(0,u)}+\delta\left(\|\omega\|_{L^\infty_{\ub}H^k(u)}\|\tr\chi\|_{L^\infty_{\ub}H^k(u)}+\|\Omega\|_{L^\infty_{\ub}H^k(u)}\|\chih,\tr\chi\|^2_{L^\infty_{\ub}H^k(u)}\right)\\
&\lesssim\mathcal{O}^{(0)}_k+\delta\Delta^2+\mathcal{O}_k[\chih]^2\lesssim C(\mathcal{O}^{(0)}_k,\mathcal{R}_k[\alpha])\end{align*}
if $\delta$ is sufficiently small depending on $\Delta$. Here we have used \eqref{estimatechih} to estimate $\mathcal{O}_k[\chih]^2$. Then we have
\begin{align}\label{estimatetrchi}
\mathcal{O}_k[\tr\chi]\lesssim C(\mathcal{O}^{(0)}_k,\mathcal{R}_k[\alpha]).\end{align}

We then turn to the equation for $\Dbh\chibh$. We will have
 \begin{align*}
\|\chibh\|_{H^k(\ub,u)}&\lesssim\|\chibh\|_{H^k(\ub,u_0)}+\|\omegab\|_{L^\infty_{u}H^k(\ub)}\|\chibh\|_{L^\infty_{u}H^k(\ub)}+\|\Omega\|_{L^\infty_uH^k(\ub)}\|\alphab\|_{L^2_{u}H^k(\ub)}\\&\lesssim\delta^{\frac{1}{2}}\mathcal{O}^{(0)}_k+\delta\Delta^2+\delta^{\frac{1}{2}}\underline{\mathcal{R}}_k[\alphab]\lesssim\delta^{\frac{1}{2}}\mathcal{O}^{(0)}_k+\delta^{\frac{1}{2}}\underline{\mathcal{R}}_k[\alphab]+\delta^{\frac{1}{2}}\end{align*}
if $\delta$ is sufficiently small depending on $\Delta$, where we use 
\begin{equation}\label{L1L2Linfty'}
\begin{split}
\|\cdot\|_{L^1_{u}H^k(\ub)}\lesssim&\|\cdot\|_{L^\infty_{u}H^k(\ub)},\\
\|\cdot\|_{L^1_{u}H^k(\ub)}\lesssim&\|\cdot\|_{L^2_{u}H^k(\ub)}
\end{split} 
\end{equation}
which are similar to \eqref{L1Linfty} and \eqref{L1L2}. Hence we have
\begin{align}\label{estimatechibh}
\mathcal{O}_k[\chibh]\lesssim \mathcal{O}^{(0)}_k+\underline{\mathcal{R}}_k[\alphab]+1.
\end{align}

The next equation for $\Db\tr\chib$ should be rewritten as a renormalized form, i.e, in terms of $\widetilde{\tr\chib}=\tr\chib+\frac{2}{|u|}$. We write
\begin{align*}
\Db\left(\tr\chib+\frac{2}{|u|}\right)&=\omegab\tr\chib-\frac{1}{2}\Omega(\tr\chib)^2-\Omega|\chibh|^2+\frac{2}{|u|^2}\\
&=-\frac{2\omegab}{|u|}+\omegab\widetilde{\tr\chib}-\Omega|\chibh|^2-\frac{1}{2}(\Omega-1)(\tr\chib)^2-\frac{1}{2}(\widetilde{\tr\chib})^2+\frac{2\widetilde{\tr\chib}}{|u|}.
\end{align*}
The last term can be absorbed by the left as
\begin{align*}
\Db\left(|u|^2\widetilde{\tr\chib}\right)=|u|^2\left(-\frac{2\omegab}{|u|}+\omegab\widetilde{\tr\chib}-\Omega|\chibh|^2-\frac{1}{2}(\Omega-1)(\tr\chib)^2-\frac{1}{2}(\widetilde{\tr\chib})^2\right).
\end{align*}
The right hand side in the norm $\|\cdot\|_{L^\infty_uH^k(\ub)}$ is estimated term by term as follows. The first terms is estimated by
\begin{align*}
\|-2|u|\omegab\|_{L^\infty_uH^k(\ub)}\lesssim\delta^{\frac{1}{2}}\mathcal{O}_k[\omegab],
\end{align*}
the second, the third and the last terms are estimated by
\begin{align*}
\left\||u|^2\left(-\omegab\widetilde{\tr\chib}-\Omega|\chibh|^2-\frac{1}{2}(\widetilde{\tr\chib})^2\right)\right\|_{L^\infty_uH^k(\ub)}\lesssim\delta\Delta^2\lesssim\delta^{\frac{1}{2}}
\end{align*}
if $\delta$ is sufficiently small depending on $\Delta$. The fourth term is estimated by
\begin{align*}
\left\||u|^2\left(-\frac{1}{2}(\Omega-1)(\tr\chib)^2\right)\right\|_{L^\infty_uH^k(\ub)}&\lesssim\|\Omega-1\|_{L^\infty_uH^k(\ub)}\left(\left\|\frac{2}{|u|}\right\|_{L^\infty_uH^k(\ub)}+\|\widetilde{\tr\chib}\|_{L^\infty_uH^k(\ub)}\right)^2\\
&\lesssim\delta^{\frac{1}{2}}\mathcal{O}^{(0)}_k(1+\delta^{\frac{1}{2}}\Delta+\delta\Delta^2)\lesssim\delta^{\frac{1}{2}}\mathcal{O}^{(0)}_k
\end{align*}
if $\delta$ is sufficiently small depending on $\Delta$. Here we use the estimate \eqref{Omega-1}. Combining the above estimates we have
\begin{align*}
\||u|^2\widetilde{\tr\chib}\|_{H^k(\ub,u)}\lesssim\||u|^2\widetilde{\tr\chib}\|_{H^k(\ub,u_0)}+\left\|\Db\left(|u|^2\widetilde{\tr\chib}\right)\right\|_{L^\infty_uH^k(\ub)}\lesssim\delta^{\frac{1}{2}}\mathcal{O}^{(0)}_k+\delta^{\frac{1}{2}}\mathcal{O}_k[\omegab]+\delta^{\frac{1}{2}}
\end{align*}
which implies that
\begin{align}\label{estimatetrchibpre}
\mathcal{O}_k[\widetilde{\tr\chib}]\lesssim\mathcal{O}^{(0)}_k+\mathcal{O}_k[\omegab]+1.
\end{align}
To complete the estimate for $\widetilde{\tr\chib}$, we need the to estimate $\omegab$. So we turn to the equation for $D\omegab$, we will have
 \begin{align*}
&\|\omegab\|_{H^k(\ub,u)}\\
\lesssim&\|\omegab\|_{H^k(0,u)}+\|\Omega\|^2_{L^\infty_{\ub}H^k(u)}\left(\delta\left(\|\eta\|_{L^\infty_{\ub}H^k(u)}\|\etab\|_{L^\infty_{\ub}H^k(u)}+\|\eta\|^2_{L^\infty_{\ub}H^k(u)}\right)+\delta^{\frac{1}{2}}\|\rho\|_{L^2_{\ub}H^k(u)}\right)\\
\lesssim&\delta^{\frac{1}{2}}\mathcal{O}^{(0)}_k+\delta^2\Delta^2+\delta\mathcal{R}_k[\rho]\lesssim\delta^{\frac{1}{2}}\mathcal{O}^{(0)}_k+\delta^{\frac{1}{2}}
\end{align*}
if $\delta$ is sufficiently small depending on $\mathcal{R}_k[\rho]$ and $\Delta$. This implies that
\begin{align}\label{estimateomegab}
\mathcal{O}_k[\omegab]\lesssim\mathcal{O}^{(0)}_k+1
\end{align}
and together with \eqref{estimatetrchibpre}, we have
\begin{align}\label{estimatetrchib}
\mathcal{O}_k[\widetilde{\tr\chib}]\lesssim\mathcal{O}^{(0)}_k+1.
\end{align}

We turn to the equation for $D\eta$. We will have
 \begin{align*}
&\|\eta\|_{H^k(\ub,u)}\\
\lesssim&\|\eta\|_{H^k(0,u)}+\|\Omega\|_{L^\infty_{\ub}H^k(u)}\left(\delta\|\chih,\tr\chi\|_{L^\infty_{\ub}H^k(u)}\|\etab\|_{L^\infty_{\ub}H^k(u)}+\delta^{\frac{1}{2}}\|\beta\|_{L^2_{\ub}H^k(u)}\right)\\
\lesssim&\delta^{\frac{1}{2}}\mathcal{O}^{(0)}_k+\delta\Delta^2+\delta^{\frac{1}{2}}\mathcal{R}_k[\beta]\lesssim\delta^{\frac{1}{2}}(\mathcal{O}^{(0)}_k+\mathcal{R}_k[\beta]+1).
\end{align*}
This implies that
\begin{align}\label{estimateeta}
\mathcal{O}_k[\eta]\lesssim\mathcal{O}^{(0)}_k+\mathcal{R}_k[\beta]+1.
\end{align}

We then consider the equation for $\Db\etab$. We will have
 \begin{align*}
\|\etab\|_{H^k(\ub,u)}&\lesssim\|\etab\|_{H^k(\ub,u_0)}+\|\Omega\|_{L^\infty_uH^k(\ub)}\left(\|\chibh,\tr\chib\|_{L^\infty_{u}H^k(\ub)}\|\eta\|_{L^\infty_{u}H^k(\ub)}+\|\betab\|_{L^2_{u}H^k(\ub)}\right)\\&\lesssim\delta^{\frac{1}{2}}\mathcal{O}^{(0)}_k+\delta^{\frac{1}{2}}\mathcal{O}_k[\eta]+\delta\Delta^2+\delta^{\frac{1}{2}}\underline{\mathcal{R}}_k[\betab]\lesssim\delta^{\frac{1}{2}}\mathcal{O}^{(0)}_k+\delta^{\frac{1}{2}}(\mathcal{R}_k[\beta]+\underline{\mathcal{R}}_k[\betab])+\delta^{\frac{1}{2}}\end{align*}
if $\delta$ is sufficiently small depending on $\Delta$. Here we use $\tr\chib=-\frac{2}{|u|}+\widetilde{\tr\chib}$ to estimate $\tr\chib$, and use \eqref{estimateeta} to estimate $\eta$. We then have
\begin{align}\label{estimateetab}
\mathcal{O}_k[\etab]\lesssim\mathcal{O}^{(0)}_k+\mathcal{R}_k[\beta]+\underline{\mathcal{R}}_k[\betab]+1.
\end{align}

At last we consider the equation for $\Db\omega$. We will have
 \begin{align*}
&\|\omega\|_{H^k(\ub,u)}\\
\lesssim&\|\omega\|_{H^k(\ub,u_0)}+\|\Omega\|^2_{L^\infty_uH^k(\ub)}\left(\|\eta\|_{L^\infty_{u}H^k(\ub)}\|\etab\|_{L^\infty_{u}H^k(\ub)}+\|\eta\|^2_{L^\infty_{u}H^k(\ub)}+\|\rho\|_{L^2_{u}H^k(\ub)}\right)\\
\lesssim&\mathcal{O}^{(0)}_k+\delta\Delta^2+\underline{\mathcal{R}}_k[\rho]\lesssim\mathcal{O}^{(0)}_k+\underline{\mathcal{R}}_k[\rho]+1\end{align*}
if $\delta$ is sufficiently small depending on  $\Delta$. Hence we have
\begin{align}\label{estimateomega}
\mathcal{O}_k[\omega]\lesssim\mathcal{O}^{(0)}_k+\underline{\mathcal{R}}_k[\rho]+1.
\end{align}

We have obtained the desired estimates \eqref{estimatechih}, \eqref{estimatetrchi}, \eqref{estimatechibh}, \eqref{estimateomegab}, \eqref{estimatetrchib}, \eqref{estimateeta}, \eqref{estimateetab}, \eqref{estimateomega}. If $\Delta$ is chosen sufficiently large (as a function of $\mathcal{O}^{(0)}_k$), we have improved the bootstrap assumption \eqref{bootstrap}. By a bootstrap argument, this implies that \eqref{bootstrap} is in fact true and then the estimates mentioned above are in fact true. Moreover, since $\Delta$ is chosen depending on $\mathcal{O}^{(0)}_k$, then $\delta$ can be chosen depending only on  $\mathcal{O}^{(0)}_k$ and $\mathcal{R}_k$. We then complete the proof of Proposition \ref{curvatureconnection}.

\end{proof}

\begin{proposition}\label{curvature}

If $\delta$ is sufficiently small depending on $\mathcal{O}^{(0)}_k$ and $\mathcal{R}^{(0)}_k$, then
$$\mathcal{R}_k, \underline{\mathcal{R}}_k\lesssim C(\mathcal{O}^{(0)}_k,\mathcal{R}^{(0)}_k)$$
where $C(\mathcal{O}^{(0)}_k, \mathcal{R}^{(0)}_k)$ is a constant depending only on $\mathcal{O}^{(0)}_k, \mathcal{R}^{(0)}_k$.
\end{proposition}

\begin{proof}
The proof is based on the following lemma:
\begin{lemma}\label{divergencetheorem} If $\delta$ is sufficiently small depending on $\mathcal{O}^{(0)}_k$ and $\mathcal{R}^{(0)}_k$, we have
\begin{align*}
\|R\|^2_{L^2_{\ub}H^k(u)}+\|\underline{R}\|^2_{L^2_uH^k(\ub)}\lesssim&\|R\|^2_{L^2_{\ub}H^k(u_0)}+\|\underline{R}\|^2_{L^2_uH^k(0)}\\&+\int_{u_0}^{u_1}\int_0^{\delta}\|\Db R+\Omega\mathcal{D}R\|_{H^k(\ub,u)}\|R\|_{H^k(\ub,u)}\D\ub\D u\\
&+\int_{u_0}^{u_1}\int_0^{\delta}\|DR-\Omega{}^*\mathcal{D}\underline{R}\|_{H^k(\ub,u)}\|\underline{R}\|_{H^k(\ub,u)}\D\ub\D u
\end{align*}
for any pair of the curvature components $(R,\underline{R})$. Here $\mathcal{D}$ is one of the Hodge operators or their $L^2$ adjoint, and ${}^*\mathcal{D}$ is the $L^2$ adjoint of $\mathcal{D}$, i.e., 
$$\int_{S_{\ub,u}}\gs(\mathcal{D} R,\underline{R})\D\mu_{\gs}=\int_{S_{\ub,u}}\gs(R,{}^*\mathcal{D} \underline{R})\D\mu_{\gs},$$
where the contractions of tensorfields by $\gs$ are defined in their natural ways.
\end{lemma}
\begin{proof}[Proof of Lemma \ref{divergencetheorem}]
We compute (in a schematic way) for $0\le i\le k$ that
\begin{align*}
\Db\nablas^iR=&\nablas^i(\Db R+\Omega\mathcal{D}\underline{R})+[\Db,\nablas^i]R-\nablas^i(\Omega\mathcal{D}\underline{R})\\
=&\nablas^i(\Db R+\Omega\mathcal{D}\underline{R})+\sum_{j=1}^i\nablas^j(\Omega\chib)\nablas^{i-j}R-\nablas^i\mathcal{D}(\Omega\underline{R})+\nabla^i(\underline{R}\cdot\nablas\Omega)\\
=&\nablas^i(\Db R+\Omega\mathcal{D}\underline{R})-\mathcal{D}\left(\nablas^i(\Omega\underline{R})\right)\\&+\sum_{j_1+j_2=i-1}\nablas^{1+j_1}(\Omega\chib)\nablas^{j_2}R+\sum_{j_1+j_2=i-1}\nablas^{j_1}K\nablas^{j_2}(\Omega\underline{R})+\nablas^i(\underline{R}\cdot \Omega(\eta+\etab)),
\end{align*}
where we use the relation 
$\nablas\Omega=\Omega\nablas\log\Omega=\frac{\Omega(\eta+\etab)}{2}$, and similarly
\begin{align*}
D\nablas^i\underline{R}=&\nablas^i(D\underline{R}-\Omega{}^*\mathcal{D}R)+{}^*\mathcal{D}\left(\nablas^i(\Omega R)\right)\\&+\sum_{j_1+j_2=i-1}\nablas^{1+j_1}(\Omega\chi)\nablas^{j_2}R-\sum_{j_1+j_2=i-1}\nablas^{j_1}K\nablas^{j_2}(\Omega R)-\nablas^i(R \cdot \Omega(\eta+\etab)).
\end{align*}
If $R$ or $\underline{R}$ is $\rho$ or $\sigma$, then we must have $j_2\ge1$ in the sum. We then compute
\begin{equation*}
\begin{split}
&\Db(|\nablas^iR|^2\D\mu_{\gs})+D(|\nablas^i\underline{R}|^2\D\mu_{\gs})\\
=&\left(\Omega\chib\cdot|\nablas^iR|^2+\Omega\chi\cdot|\nablas^i\underline{R}|^2+2\gs(\nablas^iR,\Db\nablas^iR)+2\gs(\nablas^i\underline{R},D\nablas^i\underline{R})\right)\D\mu_{\gs}\\
\end{split}
\end{equation*}
Integrating the above formula over $[0,\ub]\times[u_0,u]\times S_{\ub',u'}$ and summing over $0\le i\le k$,
\begin{equation}\label{divergence1}
\begin{split}
&\int_0^{\ub}\|R\|^2_{H^k(\ub',u)}\D\ub'+\int_{u_0}^{u}\|\underline{R}\|^2_{H^k(\ub,u')}\D u'\\
\lesssim&\int_0^{\ub}\|R\|^2_{H^k(\ub',u_0)}\D\ub'+\int_{u_0}^{u}\|\underline{R}\|^2_{H^k(0,u')}\D u'\\
&+\int_{u_0}^u\int_0^{\ub}\|\Db R+\Omega\mathcal{D}R\|_{H^k(\ub',u')}\|R\|_{H^k(\ub',u')}\D\ub'\D u'\\
&+\int_{u_0}^u\int_0^{\ub}\|DR-\Omega{}^*\mathcal{D}\underline{R}\|_{H^k(\ub',u')}\|\underline{R}\|_{H^k(\ub',u')}\D\ub'\D u'\\
&+\underbrace{\int_{u_0}^u\int_0^{\ub}\|\chib\|_{H^k(\ub',u')}\|R\|^2_{H^k(\ub',u')}\D\ub'\D u'}_{\uppercase\expandafter{\romannumeral1}}+\underbrace{\int_{u_0}^u\int_0^{\ub}\|\chi\|_{H^k(\ub',u')}\|\underline{R}\|^2_{H^k(\ub',u')}\D\ub'\D u'}_{\uppercase\expandafter{\romannumeral2}}\\
&+\underbrace{\int_{u_0}^u\int_0^{\ub}\|\eta,\etab\|_{H^k(\ub',u')}\|R\|_{H^k(\ub',u')}\|\underline{R}\|_{H^k(\ub',u')}\D\ub'\D u'}_{\uppercase\expandafter{\romannumeral3}}\\
&+\underbrace{\int_{u_0}^u\int_0^{\ub}\|K\|_{H^{k-1}(\ub',u')}\|R\|_{H^k(\ub',u')}\|\underline{R}\|_{H^k(\ub',u')}\D\ub'\D u'}_{\uppercase\expandafter{\romannumeral4}}
\end{split}
\end{equation}
where we have used both of the inequalities in \eqref{Sobolev1}. Since $\|\chib\|_{H^k(\ub,u)}\lesssim\frac{2}{|u|}+\delta^{\frac{1}{2}}(\mathcal{O}^{(0)}_k+\underline{\mathcal{R}}_k[\alphab]+1)$ from \eqref{estimatechibh}, \eqref{estimatetrchib}, we have
\begin{align}\label{estimateI}
\uppercase\expandafter{\romannumeral1}\lesssim\int_{u_0}^u\frac{2}{|u|}\int_0^{\ub}\|R\|^2_{H^k(\ub',u')}\D\ub'\D u'+\delta^{\frac{1}{2}}(\mathcal{O}^{(0)}_k+\underline{\mathcal{R}}_k[\alphab]+1)\sup_{u_0\le u'\le u}\int_0^{\ub}\|R\|^2_{H^k(\ub',u')}\D \ub'.
\end{align}
The first term on the right can be absorbed by the left hand side of \eqref{divergence1} by Gronwall inequality, and the second term on the right can be absorbed by the left hand side of \eqref{divergence1} after taking supremum over $u'\in[u_0,u]$ if $\delta$ is sufficiently small depending on $\mathcal{O}^{(0)}_k$ and $\underline{\mathcal{R}}_k$. For the term \uppercase\expandafter{\romannumeral2}, we use $\|\chi\|_{H^k(\ub,u)}\lesssim\delta^{-\frac{1}{2}}C(\mathcal{O}^{(0)}_k,\mathcal{R}_k[\alpha])$ from \eqref{estimatechih}, \eqref{estimatetrchi}, then
\begin{align}\label{estimateII}
\uppercase\expandafter{\romannumeral2}\lesssim\delta^{\frac{1}{2}}C(\mathcal{O}^{(0)}_k,\mathcal{R}_k[\alpha])\sup_{0\le\ub'\le\ub}\int_{u_0}^{u}\|\underline{R}\|^2_{H^k(\ub',u')}\D u'
\end{align}
which can be absorbed by the left hand side of \eqref{divergence1} after taking supremum over $\ub'\in[0,\ub]$ if $\delta$ is sufficiently small depending on $\mathcal{O}^{(0)}_k$ and $\mathcal{R}_k$. The term \uppercase\expandafter{\romannumeral3} is estimated as
\begin{equation}\label{estimateIII}
\begin{split}
\uppercase\expandafter{\romannumeral3}
\lesssim&\|\eta,\etab\|_{L^\infty_{\ub}L^\infty_u H^k}\sup_{u_0\le u'\le u}\left(\int_0^{\ub}\|R\|^2_{H^k(\ub',u')}\D \ub'\right)^{\frac{1}{2}}\sup_{0\le\ub'\le\ub}\left(\delta\int_{u_0}^{u}\|\underline{R}\|^2_{H^k(\ub',u')}\D u'\right)^{\frac{1}{2}}\\
\lesssim&\delta^{\frac{1}{2}}(\mathcal{O}^{(0)}_k+\mathcal{R}_k[\beta]+\underline{\mathcal{R}}_k[\betab]+1)\\&\times\delta^{\frac{1}{2}}\left(\sup_{u_0\le u'\le u}\int_0^{\ub}\|R\|^2_{H^k(\ub',u')}\D \ub'+\sup_{0\le\ub'\le\ub}\int_{u_0}^{u}\|\underline{R}\|^2_{H^k(\ub',u')}\D u'\right),
\end{split}
\end{equation}
where we use \eqref{estimateeta} and \eqref{estimateetab}. This term can also be absorbed by the left hand side of \eqref{divergence1} if $\delta$ is sufficiently small depending on $\mathcal{O}^{(0)}_k$, $\mathcal{R}_k$ and $\underline{\mathcal{R}}_k$. To estimate \uppercase\expandafter{\romannumeral4}, we need the Gauss equation, expressing the Gaussian curvature $K$ in terms of the spacetime connection and curvature,
\begin{align*}K=-\frac{1}{4}\tr\chi \tr\underline{\chi}+\frac{1}{2}(\widehat{\chi},\widehat{\underline{\chi}})-\rho.\end{align*}
To obtain an estimate for $K$, we use one of the null Bianchi equations
\begin{align*}
D\rho+\frac{3}{2}\Omega \tr\chi \rho-\Omega\{ \divs{\beta} +(2\underline{\eta}+\zeta,\beta)-\frac{1}{2}(\underline{\widehat{\chi}},\alpha) \}=0.
\end{align*}
Then by \eqref{transportestimate1}, the second of \eqref{Sobolev1}, \eqref{L1Linfty}, \eqref{L1L2} and the estimates in Proposition \ref{curvatureconnection}, we have
\begin{align*}
\|\rho\|_{H^{k-1}(\ub,u)}\lesssim&\|\rho\|_{H^{k-1}(0,u)}+\delta\|\tr\chi\|_{L^\infty_{\ub}H^k(u)}\|\rho\|_{L^\infty_{\ub}H^{k-1}(u)}+\delta^{\frac{1}{2}}\|\beta\|_{L^2_{\ub}H^k(u)}\\
&+\|\eta,\etab\|_{L^\infty_{\ub}H^k(u)}\cdot\delta^{\frac{1}{2}}\|\beta\|_{L^2_{\ub}H^{k-1}(u)}+\|\chibh\|_{L^\infty_{\ub}H^k(u)}\cdot\delta^{\frac{1}{2}}\|\alpha\|_{L^2_{\ub}H^{k-1}(u)}\\
\lesssim&\mathcal{R}^{(0)}_k+\delta C(\mathcal{O}^{(0)}_k, \mathcal{R}_k[\alpha])\|\rho\|_{L^\infty_{\ub}H^{k-1}(u)}\\
&+\delta^{\frac{1}{2}}(1+\delta^{\frac{1}{2}}(\mathcal{O}^{(0)}_k+\mathcal{R}_k[\beta]+\underline{R}_k[\betab]+1))\mathcal{R}_k[\beta]\\
&+(\mathcal{O}^{(0)}_k+\underline{\mathcal{R}}_k[\alphab]+1)\mathcal{R}_k[\alpha].
\end{align*}
If $\delta$ is sufficiently small depending on $\mathcal{O}^{(0)}_k$, $\mathcal{R}_k$ and $\underline{\mathcal{R}}_k$, then 
\begin{align*}
\|\rho\|_{H^{k-1}(\ub,u)}\lesssim\mathcal{R}^{(0)}_k+(\mathcal{O}^{(0)}_k+\underline{\mathcal{R}}_k[\alphab]+1)\mathcal{R}_k[\alpha]+1
\end{align*}
and therefore by the Gauss equation and the estimates \eqref{estimatechih}, \eqref{estimatetrchi}, \eqref{estimatetrchi}, \eqref{estimatetrchib},
\begin{align*}
\|K\|_{H^{k-1}(\ub,u)}\lesssim C(\mathcal{O}^{(0)}_k, \mathcal{R}_k,\underline{\mathcal{R}}_k).
\end{align*}
Then the term \uppercase\expandafter{\romannumeral4} is estimated by
\begin{align*}
\uppercase\expandafter{\romannumeral4}
\lesssim&\|K\|_{L^\infty_{\ub}L^\infty_u H^{k-1}}\sup_{u_0\le u'\le u}\left(\int_0^{\ub}\|R\|^2_{H^k(\ub',u')}\D \ub'\right)^{\frac{1}{2}}\sup_{0\le\ub'\le\ub}\left(\delta\int_{u_0}^{u}\|\underline{R}\|^2_{H^k(\ub',u')}\D u'\right)^{\frac{1}{2}}\\
\lesssim&\delta^{\frac{1}{2}}C(\mathcal{O}^{(0)}_k, \mathcal{R}_k, \underline{\mathcal{R}}_k)\left(\sup_{u_0\le u'\le u}\int_0^{\ub}\|R\|^2_{H^k(\ub',u')}\D \ub'+\sup_{0\le\ub'\le\ub}\int_{u_0}^{u}\|\underline{R}\|^2_{H^k(\ub',u')}\D u'\right),
\end{align*}
which can be absorbed by the left hand side of \eqref{divergence1} if $\delta$ is sufficiently small depending on $\mathcal{O}^{(0)}_k, \mathcal{R}_k, \underline{\mathcal{R}}_k$. The proof of Lemma \ref{divergencetheorem} is then completed by choosing $\ub=\delta$ and $u=u_1$.

\end{proof}

The estimates for the curvature components are done through the null Bianchi equations which can also be found in \cite{Chr08}. They can be written in four groups as follows.
\begin{align*}
&\begin{cases}
\underline{\widehat{D}}\alpha-\Omega\nablas\tensor\beta = \underbrace{\frac{1}{2}\Omega \tr\underline{\chi}\alpha - 2\underline{\omega} \alpha}_{\uppercase\expandafter{\romannumeral1}} -\Omega\{\underbrace{ -(4\eta+\zeta)\widehat{\otimes}\beta}_{\uppercase\expandafter{\romannumeral3}} +3\widehat{\chi}\rho+3^\ast\widehat{\chi}\sigma\}
\\D\beta-\Omega\divs\alpha=\underbrace{-\frac{3}{2}\Omega \tr\chi \beta+\Omega \widehat{\chi}\cdot\beta-\omega\beta}_{\uppercase\expandafter{\romannumeral2}}+\underbrace{\Omega (\underline{\eta}+2\zeta)\cdot\alpha}_{\uppercase\expandafter{\romannumeral3}}
\end{cases}\\
&\begin{cases}
\underline{D}\underline{\beta}-\Omega\divs\alphab=\underbrace{-\frac{3}{2}\Omega \tr\underline{\chi} \underline{\beta}+{\Omega} \widehat{\underline{\chi}}\cdot\underline{\beta}+\underline{\omega}\underline{\beta}}_{\uppercase\expandafter{\romannumeral1}}-\underbrace{\Omega ({\eta}+2\zeta)\cdot\underline{\alpha}}_{\uppercase\expandafter{\romannumeral3}}
\\
\widehat{D}\underline{\alpha} +\Omega\nablas\tensor\betab=\underbrace{ \frac{1}{2}\Omega \tr\chi\underline{\alpha} - 2\omega\underline{\alpha}}_{\uppercase\expandafter{\romannumeral2}} -\Omega\{\underbrace{(4\underline{\eta}-\zeta)\widehat{\otimes}\underline{\beta } }_{\uppercase\expandafter{\romannumeral3}}+3\underline{\widehat{\chi}}\rho-3^\ast\underline{\widehat{\chi}}\sigma\}
\end{cases}\\
&\begin{cases}
\underline{D}\beta- \Omega( \ds{\rho} +^*\ds{\sigma})= \underbrace{- \frac{1}{2}\Omega \tr\underline{\chi}\beta+\Omega\underline{\widehat{\chi}} \cdot\beta-\underline{\omega}\beta}_{\uppercase\expandafter{\romannumeral1}}+\Omega\{\underbrace{3\eta\rho+3^*\eta\sigma}_{\uppercase\expandafter{\romannumeral3}} +2\widehat{\chi}^\sharp\cdot\underline{\beta}\}
\\
D\rho-\Omega\divs\beta=\underbrace{-\frac{3}{2}\Omega \tr\chi \rho}_{\uppercase\expandafter{\romannumeral2}}+\Omega\{ \underbrace{(2\underline{\eta}+\zeta,\beta)}_{\uppercase\expandafter{\romannumeral3}}\boxed{-\frac{1}{2}(\underline{\widehat{\chi}},\alpha)} \}
\\
D\sigma+\Omega\curls\beta=\underbrace{-\frac{3}{2}\Omega \tr\chi\sigma}_{\uppercase\expandafter{\romannumeral2}}-\Omega\{\underbrace{(2\underline{\eta}+\zeta),^*\beta)}_{\uppercase\expandafter{\romannumeral3}}\boxed{-\frac{1}{2}\widehat{\underline{\chi}}\wedge\alpha}\}
\end{cases}
\\
&\begin{cases}
\underline{D}\rho+\Omega\divs\betab=\underbrace{-\frac{3}{2}\Omega \tr\underline{\chi} \rho}_{\uppercase\expandafter{\romannumeral1}}-\Omega\{ \underbrace{(2\eta-\zeta,\underline{\beta})}_{\uppercase\expandafter{\romannumeral3}}\boxed{+\frac{1}{2}(\widehat{\chi},\underline{\alpha})   }  \}
\\
\underline{D}\sigma+\Omega\curls\betab=\underbrace{-\frac{3}{2}\Omega \tr\underline{\chi}\sigma}_{\uppercase\expandafter{\romannumeral1}}-\Omega\{\underbrace{(2\underline{\eta}-\zeta),^*\underline{\beta})}_{\uppercase\expandafter{\romannumeral3}}\boxed{-\frac{1}{2}\widehat{{\chi}}\wedge\underline{\alpha}}\} 
\\
D\underline{\beta }+ \Omega(\ds{\rho} -^*\ds{\sigma})=\underbrace{- \frac{1}{2}\Omega \tr\chi\underline{\beta}+\Omega\widehat{\chi}\cdot\underline{\beta}-\omega\underline{\beta}}_{\uppercase\expandafter{\romannumeral2}}-\Omega\{ \underbrace{3\etab\rho-3^*\underline{\eta}\sigma}_{\uppercase\expandafter{\romannumeral3}} -2\underline{\widehat{\chi}}^\sharp\cdot\beta\}
\end{cases}
\end{align*}
We apply Lemma \ref{divergencetheorem} for four cases:
$$(R,\underline{R})=\left(\frac{\alpha}{\sqrt{2}}, \beta\right), \left(-\betab,\frac{\alphab}{\sqrt{2}}\right), \left(\beta,(-\rho,\sigma)\right), \left((\rho,\sigma),-\betab\right)$$
where the corresponding $\mathcal{D}$ and ${}^*\mathcal{D}$ are
\begin{itemize}
\item $\mathcal{D}\beta=-\frac{1}{\sqrt{2}}\nablas\tensor\beta$, ${}^*\mathcal{D}\left(\frac{\alpha}{\sqrt{2}}\right)=\sqrt{2}\divs\left(\frac{\alpha}{\sqrt{2}}\right)$;
\item $\mathcal{D}\left(\frac{\alpha}{\sqrt{2}}\right)=-\sqrt{2}\divs\left(\frac{\alphab}{\sqrt{2}}\right)$, ${}^*\mathcal{D}(-\betab)=\frac{1}{\sqrt{2}}(-\betab)$;
\item $\mathcal{D}(-\rho,\sigma)=-\nablas\rho-{}^*\nablas\sigma$, ${}^*\mathcal{D}\beta=(-\divs\beta,-\curls\beta)$;
\item $\mathcal{D}(-\betab)=(\divs\betab,\curls\betab)$, ${}^*\mathcal{D}(\rho,\sigma)=-\nablas\rho+{}^*\nablas\sigma$.
\end{itemize}
What we need to do is to plug the right hand side of the above null Bianchi equations in the following two integrals:
\begin{align*}\int_{u_0}^{u_1}\int_0^{\delta}\|\Db R+\Omega\mathcal{D}R\|_{H^k(\ub',u')}\|R\|_{H^k(\ub',u')}\D\ub'\D u',\\
\int_{u_0}^{u_1}\int_0^{\delta}\|DR-\Omega{}^*\mathcal{D}\underline{R}\|_{H^k(\ub',u')}\|\underline{R}\|_{H^k(\ub',u')}\D\ub'\D u'.\end{align*}
Because the estimates \eqref{estimateomegab}, \eqref{estimateomega} for $\omegab$, $\omega$ have the same power of $\delta$ as $\widetilde{\tr\chib}$ and $\tr\chi$ respectively, the terms labeled by \uppercase\expandafter{\romannumeral1}, \uppercase\expandafter{\romannumeral2} and \uppercase\expandafter{\romannumeral3} on the right hand side of the above null Bianchi equations, can be estimated in the same way as what we do to \uppercase\expandafter{\romannumeral1}, \uppercase\expandafter{\romannumeral2} and \uppercase\expandafter{\romannumeral3} in \eqref{divergence1} respectively, i.e., \eqref{estimateI}, \eqref{estimateII} and \eqref{estimateIII}. So these terms can also be absorbed by the left hand side, and we only need to concern in the following the terms which are not labeled. We then have the following four estimates: For $(R,\underline{R})=\left(\frac{\alpha}{\sqrt{2}}, \beta\right)$, we have
\begin{align*}
\|\alpha\|^2_{L^2_{\ub}H^k(u)}+\|\beta\|^2_{L^2_uH^k(\ub)}\lesssim&\|\alpha\|^2_{L^2_{\ub}H^k(u_0)}+\|\beta\|^2_{L^2_uH^k(0)}\\
&+\int_{u_0}^{u_1}\int_0^{\delta}\|\chih\|_{H^k(\ub',u')}\|\rho,\sigma\|_{H^k(\ub',u')}\|\beta\|_{H^k(\ub',u')}\D\ub'\D u'\\
\lesssim&\delta^{-2}(\mathcal{R}^{(0)}_k)^2+\|\chih\|_{L^\infty_{\ub}L^\infty_uH^k}\|\rho,\sigma\|_{L^\infty_uL^2_{\ub}H^k}\|\beta\|_{L^\infty_uL^2_{\ub}H^k}\\
\lesssim&\delta^{-2}(\mathcal{R}^{(0)}_k)^2+(\mathcal{O}^{(0)}_k+\mathcal{R}_k[\alpha]+1)\mathcal{R}_k[\rho,\sigma]\mathcal{R}_k[\beta]\\
\lesssim&\delta^{-2}((\mathcal{R}^{(0)}_k)^2+1)
\end{align*}
if $\delta$ is sufficiently small depending on $\mathcal{O}^{(0)}_k$ and $\mathcal{R}_k$. This immediately implies that
\begin{align}\label{estimatealphabeta}
\mathcal{R}_k[\alpha]+\underline{\mathcal{R}}_k[\betab]\lesssim\mathcal{R}^{(0)}_k+1.
\end{align}

For $(R,\underline{R})=\left(-\betab,\frac{\alphab}{\sqrt{2}}\right)$, we have
\begin{align*}
\|\betab\|^2_{L^2_{\ub}H^k(u)}+\|\alphab\|^2_{L^2_uH^k(\ub)}\lesssim&\|\betab\|^2_{L^2_{\ub}H^k(u_0)}+\|\alphab\|^2_{L^2_uH^k(0)}\\
&+\int_{u_0}^{u_1}\int_0^{\delta}\|\chibh\|_{H^k(\ub',u')}\|\rho,\sigma\|_{H^k(\ub',u')}\|\alphab\|_{H^k(\ub',u')}\D\ub'\D u'\\
\lesssim&\delta(\mathcal{R}^{(0)}_k)^2+\|\chibh\|_{L^\infty_{\ub}L^\infty_uH^k}\|\rho,\sigma\|_{L^\infty_uL^2_{\ub}H^k}\cdot\delta^{\frac{1}{2}}\|\alphab\|_{L^\infty_{\ub}L^2_{u}H^k}\\
\lesssim&\delta(\mathcal{R}^{(0)}_k)^2+\delta^2(\mathcal{O}^{(0)}_k+\underline{\mathcal{R}}_k[\alphab]+1)\mathcal{R}_k[\rho,\sigma]\underline{\mathcal{R}}_k[\alphab]\\
\lesssim&\delta((\mathcal{R}^{(0)}_k)^2+1)
\end{align*}
if $\delta$ is sufficiently small depending on $\mathcal{O}^{(0)}_k$, $\mathcal{R}_k$ and $\underline{\mathcal{R}}_k$. This implies that
\begin{align}\label{estimatebetabalphab}
\mathcal{R}_k[\betab]+\underline{\mathcal{R}}_k[\alphab]\lesssim\mathcal{R}^{(0)}_k+1.
\end{align}

For the next two cases, we will encounter the so called \emph{borderline terms}, which are the boxed terms in the equations. For $(R,\underline{R})= \left(\beta,(-\rho,\sigma)\right)$, we have
\begin{align*}
&\|\beta\|^2_{L^2_{\ub}H^k(u)}+\|\rho,\sigma\|^2_{L^2_uH^k(\ub)}\\\lesssim&\|\beta\|^2_{L^2_{\ub}H^k(u_0)}+\|\rho,\sigma\|^2_{L^2_uH^k(0)}\\
&+\int_{u_0}^{u_1}\int_0^{\delta}\|\chih\|_{H^k(\ub',u')}\|\betab\|_{H^k(\ub',u')}\|\beta\|_{H^k(\ub',u')}\D\ub'\D u'\\
&+\int_{u_0}^{u_1}\int_0^{\delta}\|\chibh\|_{H^k(\ub',u')}\|\alpha\|_{H^k(\ub',u')}\|\rho,\sigma\|_{H^k(\ub',u')}\D\ub'\D u'\\
\lesssim&(\mathcal{R}^{(0)}_k)^2+\|\chih\|_{L^\infty_{\ub}L^\infty_uH^k}\|\betab\|_{L^\infty_{\ub}L^2_uH^k}\|\beta\|_{L^\infty_{u}L^2_{\ub}H^k}\\
&+\|\chibh\|_{L^\infty_{\ub}L^\infty_uH^k}\|\alpha\|_{L^\infty_{u}L^2_{\ub}H^k}\cdot\delta^{\frac{1}{2}}\|\rho,\sigma\|_{L^\infty_{\ub}L^2_{u}H^k}\\
\lesssim&(\mathcal{R}^{(0)}_k)^2+\delta^\frac{1}{2}(\mathcal{O}^{(0)}_k+\mathcal{R}_k[\alpha]+1)\underline{\mathcal{R}}_k[\betab]\mathcal{R}_k[\beta]+(\mathcal{O}^{(0)}_k+\underline{\mathcal{R}}_k[\alphab]+1)\mathcal{R}_k[\alpha]\underline{\mathcal{R}}_k[\rho,\sigma]\\
\lesssim&(\mathcal{R}^{(0)}_k)^2+(\mathcal{O}^{(0)}_k+\underline{\mathcal{R}}_k[\alphab]+1)\mathcal{R}_k[\alpha]\underline{\mathcal{R}}_k[\rho,\sigma]+1.
\end{align*}
if $\delta$ is sufficiently small depending on $\mathcal{O}^{(0)}_k$ and $\mathcal{R}_k$. The \emph{borderline terms} give rise to the cubic term $\underline{\mathcal{R}}_k[\alphab]\mathcal{R}_k[\alpha]\underline{\mathcal{R}}_k[\rho,\sigma]$ on the right hand side no matter how small we choose $\delta$.  Fortunately, we can plug in the estimates \eqref{estimatealphabeta} and \eqref{estimatebetabalphab} derived previously, to derive
\begin{align*}
\mathcal{R}_k[\beta]^2+\underline{\mathcal{R}}_k[\rho,\sigma]^2\lesssim(\mathcal{R}^{(0)}_k)^2+(\mathcal{O}^{(0)}_k+\mathcal{R}^{(0)}_k+1)^2\underline{\mathcal{R}}_k[\rho,\sigma]+1.
\end{align*}
This inequality is sublinear in $\underline{\mathcal{R}}_k[\rho,\sigma]$ and hence we have
\begin{align}\label{estimatebetarhosigma}
\mathcal{R}_k[\beta]+\underline{\mathcal{R}}_k[\rho,\sigma]\lesssim(\mathcal{O}^{(0)}_k)^2+(\mathcal{R}^{(0)}_k)^2+1.
\end{align}

For the final case $(R,\underline{R})= \left((\rho,\sigma),-\betab\right)$, we have
\begin{align*}
&\|\rho,\sigma\|^2_{L^2_{\ub}H^k(u)}+\|\betab\|^2_{L^2_uH^k(\ub)}\\\lesssim&\|\rho,\sigma\|^2_{L^2_{\ub}H^k(u_0)}+\|\betab\|^2_{L^2_uH^k(0)}\\
&+\int_{u_0}^{u_1}\int_0^{\delta}\|\chih\|_{H^k(\ub',u')}\|\alphab\|_{H^k(\ub',u')}\|\rho,\sigma\|_{H^k(\ub',u')}\D\ub'\D u'\\
&+\int_{u_0}^{u_1}\int_0^{\delta}\|\chibh\|_{H^k(\ub',u')}\|\beta\|_{H^k(\ub',u')}\|\betab\|_{H^k(\ub',u')}\D\ub'\D u'\\
\lesssim&(\mathcal{R}^{(0)}_k)^2+\|\chih\|_{L^\infty_{\ub}L^\infty_uH^k}\cdot\delta^{\frac{1}{2}}\|\alphab\|_{L^\infty_{\ub}L^2_{u}H^k}\cdot\delta^{\frac{1}{2}}\|\rho,\sigma\|_{L^\infty_{\ub}L^2_{u}H^k}\\
&+\|\chibh\|_{L^\infty_{\ub}L^\infty_uH^k}\|\beta\|_{L^\infty_{u}L^2_{\ub}H^k}\|\betab\|_{L^\infty_{\ub}L^2_{u}H^k}\\
\lesssim&\delta(\mathcal{R}^{(0)}_k)^2+\delta(\mathcal{O}^{(0)}_k+\mathcal{R}_k[\alpha]+1)\underline{\mathcal{R}}_k[\alphab]\underline{\mathcal{R}}_k[\rho,\sigma]+\delta^\frac{3}{2}(\mathcal{O}^{(0)}_k+\underline{\mathcal{R}}_k[\alphab]+1)\mathcal{R}_k[\beta]\underline{\mathcal{R}}_k[\betab]\\
\lesssim&\delta\left((\mathcal{O}^{(0)}_k)^2+(\mathcal{R}^{(0)}_k)^2+1\right)^2
\end{align*}
if $\delta$ is sufficiently small depending on $\mathcal{O}^{(0)}_k$ and $\mathcal{R}_k$ and we have used the estimates \eqref{estimatealphabeta}, \eqref{estimatebetabalphab} and \eqref{estimatebetarhosigma} derived in the above three cases. Then we have
\begin{align}\label{estimaterhosigmabetab}
\mathcal{R}_k[\rho,\sigma]+\underline{\mathcal{R}}_k[\betab]\lesssim(\mathcal{O}^{(0)}_k)^2+(\mathcal{R}^{(0)}_k)^2+1.
\end{align}

By the estimates \eqref{estimatealphabeta}, \eqref{estimatebetabalphab}, \eqref{estimatebetarhosigma}, \eqref{estimaterhosigmabetab}, we know that the the choice of $\delta$ in the course of the proof can be made to depend only on the initial quantities. We then complete the proof of Proposition \ref{curvature}.

\end{proof}

The proof of Theorem \ref{existence} is also completed by routine construction of solution, see for example Chapter 16, especially Chapter 16.3 of \cite{Chr08}.

\end{proof}

\subsection{The hidden smallness}
In the spacetime constructed above, $\alpha$, $\beta$ and $\chih$ are bounded by $\delta$ to some negative power, and $\rho$, $\sigma$ and $\tr\chi$ are bounded by $1$. In this subsection, we will prove that if these quantities are close to their values in a Schwarzschild spacetime, then they remain so in the spacetime. In a Schwarzschild spacetime, the only nonzero connection coefficients and curvature components are $\tr\chi$, $\tr\chib$, $\omega$, $\omegab$ and $\rho$. We use $\tr\chi_{m}$, $\tr\chib_{m}$, $\omega_{m}$, $\omegab_{m}$, $\rho_{m}$ to denote their values in the  Schwarzschild spacetime with mass $m$. Moreover,  $\omegab_{m}$ can be chosen to be zero by choosing $\Omega_{m}\equiv 1$ and $u$ can be chosen to be $-r$ on a single null cone. Consequently, on this null cone, $\tr\chib_{m}=-\frac{2}{|u|}$ and $\tr\chi_{m}=\frac{2}{|u|}-\frac{4m}{|u|^2}$. Moreover, $\rho_{m}=-\frac{2m}{r^3}$ in the Schwarzschild spacetime. We will prove

\begin{proposition}\label{prop:smallness}
In addition to the assumptions in Theorem \ref{existence}, suppose moreover that the following estimates hold for some $\ub\in [0,\delta]$:
\begin{align}\label{additioncondition}
\left\|\chih,\tr\chi-\left(\frac{2}{|u|}-\frac{4m}{|u|^2}\right),\omega-\omega_{m},\alpha,\rho+\frac{2m}{|u|^3}\right\|_{H^k(\ub,u_0)}\le\delta^{\frac{1}{2}}C
\end{align}
for some constant $C$ independent of $\delta$, some real number $m$ (allowing to be zero or even negative) and some sufficiently large integer $k$ (which may be different from that in Theorem \ref{existence}). Then the same estimates  (for a smaller $k$ and a larger $C$) together with the estimates
\begin{align*}
\left\|\beta,\sigma\right\|_{H^k(\ub,u)}\le\delta^{\frac{1}{2}}C
\end{align*}
hold on $\Cb_{\ub}$.
\end{proposition}
This theorem was already proved in \cite{L-Y}. For the sake of completeness, we still collect the proof here in a brief way.
\begin{proof}
The proof relies on integrating the null structure and Bianchi equations for $\Db\chih$, $\Db\tr\chi$, $\Db\omega$, $\Db\alpha$, $\Db\rho$, and the Gauss-Codazzi-Ricci equations. These equations can also be found in \cite{Chr08}. Consider the equation for $\Db\chih$,\footnote{ The notation $\Dbh\chih$ refers to the trace-free part of $\Db\chih$. For trace-free two-tensor $\theta$, $\Db\theta=\Dbh\theta+\frac{1}{2}\tr(\Db\theta)\gs=\Dbh\theta+(\Omega\chibh,\theta)\gs$. So one only needs to consider $\Dbh\theta$ instead of $\Db\theta$. We will also use $\Dh\theta$ to denote the trace-free part of $D\theta$ below.}
\begin{align*}
\widehat{\underline{D}}\widehat{{\chi}}=-\omegab\chih+\Omega(\nablas\widehat{\otimes}{\eta}
+{\eta}\widehat{\otimes}{\eta}+\frac{1}{2}\tr\underline{\chi} \widehat{{\chi}}-\frac{1}{2}\tr {\chi}\widehat{\underline{\chi}}).
\end{align*}
Since the estimates of Theorem \ref{existence} hold on $\Cb_{\ub}$, we will have, absorbing the first and the fourth terms by Gronwall's inequality,
\begin{equation}\label{chiimproved}
\begin{split}
\|\chih\|_{H^k(\ub,u)}\lesssim&\|\chih\|_{H^k(\ub,u_0)}+\|\eta\|_{L^\infty_u H^{k+1}(\ub)}\\&+\|\eta\|^2_{L^\infty_u H^{k}(\ub)}+\|\tr\chi\|_{L^\infty_u H^{k}(\ub)}\|\chibh\|_{L^\infty_u H^{k}(\ub)}\\\le&\delta^{\frac{1}{2}}C.
\end{split}
\end{equation}
The above estimate is true by choosing a smaller $k$ than that in Theorem \ref{existence}.

Consider the equation for $\Db\rho$, 
\begin{align*}
\underline{D}\rho+\frac{3}{2}\Omega \tr\underline{\chi} \rho+\Omega\{ \divs{\underline{\beta} } +(2\eta-\zeta,\underline{\beta})+\frac{1}{2}(\widehat{\chi},\underline{\alpha})     \}=0.
\end{align*}
The Schwarzschidean value of $\rho$ satisfies
\begin{align*}
\Db\rho_{m}+\frac{3}{2}\tr\chib_{m}\rho_{m}=0
\end{align*}
where $\rho_{m}=-\frac{2m}{|u|^3}$ and $\tr\chib_{m}=-\frac{2}{|u|}$. Then we write
\begin{align*}
&\Db(\rho-\rho_{m})+\frac{3}{2}\Omega\tr\chib(\rho-\rho_{m})\\&+\frac{3}{2}\rho_{m}(\Omega\tr\chib-\tr\chi_{m})+\Omega\{ \divs{\underline{\beta} } +(2\eta-\zeta,\underline{\beta})+\frac{1}{2}(\widehat{\chi},\underline{\alpha})     \}=0.
\end{align*}
From the estimates in Theorem \ref{existence},
\begin{align*}
\|\Omega\tr\chib-\tr\chi_{m}\|_{L^\infty_u H^k(\ub)}, \|\Omega\{\divs\betab+(2\eta-\zeta,\betab)\}\|_{L^2_u H^k(\ub)}\le \delta^{\frac{1}{2}}C
\end{align*}
and from the improved estimate \eqref{chiimproved} of $\chih$ on $\Cb_{\ub}$ above,
\begin{align*}
\|\chih\cdot\alphab\|_{L_u^2H^k(\ub,u)}\le\delta C,
\end{align*}
we will have
\begin{equation}\label{rhoimproved}\|\rho-\rho_{m}\|_{H^k(\ub,u)}\lesssim\|\rho-\rho_{m}\|_{H^k(\ub,u_0)}+\delta^{\frac{1}{2}}C'\le\delta^{\frac{1}{2}}C.\end{equation}

Consider the equation for $\Db\tr\chi$,
\begin{align*}
\underline{D}\tr{\chi}=-\omegab\tr\chi+\Omega(2\divs {\eta}+2|{\eta}|^2-(\chih,\chibh)-\frac{1}{2}\tr\chi\tr\chib+2\rho),\end{align*}
and the Schwarzschildean value of $\tr\chi$ satisfies
\begin{align*}\Db\tr\chi_{m}=-\frac{1}{2}\tr\chi_{m}\tr\chib_{m}+2\rho_{m}\end{align*}
where $\tr\chi_{m}=\frac{2}{|u|}-\frac{4m}{|u|^2}$. So we write
\begin{align*}
\Db(\tr\chi-\tr\chi_{m})
=&-\frac{1}{2}\Omega\tr\chib(\tr\chi-\tr\chi_{m})+2(\Omega \rho-\rho_{m})\\&-\omegab\tr\chi+2\Omega\divs\eta+2\Omega|\eta|^2+\frac{1}{2}\tr\chi_{m}(\Omega\tr\chib-\Omega\tr\chib_{m}).
\end{align*}
By integrating, the first term on the right can be absorbed, the second term can be estimated by \eqref{rhoimproved}, and the terms of the second line can be estimated using Theorem \ref{existence}. Consequently, we have
\begin{align}\label{trchiimproved}
\|\tr\chi-\tr\chi_{m}\|_{H^k(\ub,u)}\le \delta^{\frac{1}{2}}C.
\end{align}

Consider the equation for $\Db\omega$, 
\begin{align*}
&\underline{D}\omega=\Omega^2(2(\eta,\underline{\eta})-|\underline{\eta}|^2-\rho),
\end{align*}
which can be written as
\begin{align*}
\Db(\omega-\omega_{m})=\Omega^2(2(\eta,\underline{\eta})-|\underline{\eta}|^2)-(\Omega^2\rho-\rho_{m}).
\end{align*}
Use again the estimates in Theorem \ref{existence} and \eqref{rhoimproved},  we have
\begin{align*}
\|\omega-\omega_{m}\|_{H^k(\ub,u)}\le \delta^{\frac{1}{2}}C.
\end{align*}

Consider the null Bianchi equation for $\Db\alpha$,
\begin{align*}
\underline{\widehat{D}}\alpha = \frac{1}{2}\Omega \tr\underline{\chi}\alpha - 2\underline{\omega} \alpha +\Omega\nablas\tensor\beta +\Omega(4\eta+\zeta)\widehat{\otimes}\beta-3\Omega\widehat{\chi}\rho-3\Omega^\ast\widehat{\chi}\sigma.
\end{align*}
The first and second terms on the right hand side can be absorbed, the third and fourth terms can be estimated using Theorem \ref{existence}, the last two terms can be estimated using \eqref{chiimproved}. Finally we have
\begin{align*}
\|\alpha\|_{H^k(\ub,u)}\le\delta^{\frac{1}{2}}C.
\end{align*}

Finally, the estimates for $\beta$ and $\sigma$ can be obtained directly from the following Codazzi-Ricci equations
\begin{align*}
&\beta=-\divs \widehat{\chi}+\frac{1}{2}\ds \tr\chi-\widehat{\chi}\cdot\zeta+\frac{1}{2}\tr\chi\zeta,\\
&\sigma=\curls \eta+\frac{1}{2}\widehat{\chi}\wedge\widehat{\underline{\chi}}
\end{align*}
by using the estimates in Theorem \ref{existence} and \eqref{chiimproved}, \eqref{trchiimproved}.

\end{proof}

From Proposition \ref{prop:smallness} and Theorem \ref{existence}, we know that all connection coefficients and curvature components  are close to their values in the Schwarzschild spacetime with mass $m_0$ in $H^k(\ub,u)$ norm, if they satisfy \eqref{additioncondition} initially.

\subsection{Proofs of Theorems \ref{step1}, \ref{step3} and \ref{thm:closetoM}}

We are in the position to finish the proofs of Theorems \ref{step1}, \ref{step3} and \ref{thm:closetoM}.
\begin{proof}[Proof of Theorem \ref{step1}]
This was also essentially proved in \cite{L-Y}. Let us collect the proof briefly here. Following the construction of the initial data on $C_{u_0}$ in Section \ref{section:step1}, from Chapter 2 in \cite{Chr08}, the ansatz \eqref{shortpulse} implies that the data induced on the part $0\le\ub\le\delta$ of $C_{u_0}$ and $\Cb_0$, the Minkowskian null cone, satisfies the assumption of Theorem \ref{existence}\footnote{In fact, the data induced on the part $0\le\ub\le\delta$ of $C_{u_0}$ satisfies a more restrictive short pulse hierarchy, see the discussions in Section \ref{moreonregionIII}.} for a sufficiently large integer $k$. Moreover, from \cite{L-Y}, the condition \eqref{integrate=m0} together with the fact that $\chih=\alpha=0$ at $\ub=\delta$ and $\Omega\equiv1$ along $C_{u_0}$ imply that the data satisfies \eqref{additioncondition} for $\ub=\delta$ and $m=m_0$. Let us make the latter more clear. From the expression
\begin{align*}
\widehat{\gs}(\ub,\vartheta)_{AB}=\frac{|u_0|^2}{(1+\frac{1}{4}|\vartheta|^2)}\exp\left(\frac{\delta^{\frac{1}{2}}}{|u_0|}\psi_0(\frac{\ub}{\delta},\vartheta)\right),
\end{align*}
we know that
\begin{align*}
|\chih|^2=&\frac{1}{4}(\widehat{\gs}^{-1})^{AC}(\widehat{\gs}^{-1})^{BD}\frac{\partial \widehat{\gs}_{AB}}{\partial \ub}\frac{\partial \widehat{\gs}_{CD}}{\partial \ub}\\
=&\frac{1}{4}(I-O(\delta^{\frac{1}{2}}))^{AC}(I-O(\delta^{\frac{1}{2}}))^{BD}\frac{\delta^{-\frac{1}{2}}}{|u_0|}\left(\frac{\partial\psi_0}{\partial\ub}+O(\delta^{\frac{1}{2}})\right)_{AB}\frac{\delta^{-\frac{1}{2}}}{|u_0|}\left(\frac{\partial\psi_0}{\partial\ub}+O(\delta^{\frac{1}{2}})\right)_{CD}\\
=&\frac{1}{4}\frac{\delta^{-1}}{|u_0|^2}\left|\frac{\partial\psi_0}{\partial\ub}\right|^2+O(\delta^{-\frac{1}{2}})
\end{align*}
where $O(\delta^{c})$ is $\delta^{c}$ multiplying an expression in terms of $|u_0|$ and $\psi_0$ together with its derivatives. In particular, $\frac{\partial}{\partial\vartheta}O(\delta^{c})=O(\delta^c)$. Then
\begin{align*}
\int_0^\delta|\chih|^2\D\ub=\frac{1}{4|u_0|^2}\int_0^1\left|\frac{\partial\psi_0}{\partial s}\right|^2\D s+O(\delta^{\frac{1}{2}})
\end{align*}
and 
\begin{align*}
\tr\chi\big|_{S_{\delta,u_0}}-\frac{2}{|u_0|}=&-\frac{1}{2}\int_0^\delta(\tr\chi)^2\D\ub-\int_0^\delta|\chih|^2\D\ub\\
=&-\frac{1}{2}\int_0^\delta(\tr\chi)^2\D\ub-\frac{4m_0}{|u_0|^2}+O(\delta^{\frac{1}{2}}).
\end{align*}
This implies that
\begin{align*}
\|\tr\chi-\tr\chi_{m_0}\|_{H^k(\delta,u_0)}\le\delta^{\frac{1}{2}}C.
\end{align*}
It remains to consider $\rho$ to verify the condition \eqref{additioncondition}. It follows from writing the Gauss equation
\begin{align*}
-\rho=K+\frac{1}{4}\tr\chi\tr\chib-\frac{1}{2}(\chih,\chibh)
\end{align*}
as
\begin{align*}
-\rho-\frac{2m_0}{|u_0|^3}=K-\frac{1}{|u_0|^2}+\frac{1}{4}\tr\chi\tr\chib-\frac{1}{4}\cdot\left(\frac{2}{|u_0|}-\frac{4m_0}{|u_0|^2}\right)\cdot\left(-\frac{2}{|u_0|}\right)-\frac{1}{2}(\chih,\chibh).
\end{align*}
We should use the fact that
\begin{align*}
\left\|K-\frac{1}{|u_0|^2}\right\|_{H^k(\ub,u_0)}\le\delta^{\frac{1}{2}}C
\end{align*}
for all $0\le\ub\le\delta$ which follows easily from the equation
\begin{align*}
DK+\Omega\tr\chi K=\divs\divs(\Omega\chih)-\frac{1}{2}\Deltas(\Omega\tr\chi).
\end{align*}
We have then finished verifying condition \eqref{additioncondition}.

By applying Theorem \ref{existence}, if $\delta$ is sufficiently small, the smooth solution $g$ of the vacuum Einstein equations exists in the region $0\le\ub\le\delta$, $u_0\le u\le u_1$. By applying Proposition \ref{prop:smallness}, the data induced on $\Cb_{\delta}$ satisfies the property that all connection coefficients and curvature components are $\delta^{\frac{1}{2}}$-close to their values in the Schwarzschild spacetime $g_{m_0}$ (defined in \eqref{gm0}) with mass $m_0$ in $H^k$. Now let us fix an $\varepsilon_0>0$ and consider the part $\delta\le\ub\le\delta+\varepsilon_0$ of $C_{u_0}$. We claim that if $\delta$ is sufficiently small, the data induced on this part of $C_{u_0}$ have the same property: all connection coefficients and curvature components are $\delta^{\frac{1}{2}}$-close to their values in the Schwarzschild spacetime in the mass $m_0$ in $H^k(\delta,u)$. The proof of this claim can be done by noting that $\chih=\alpha\equiv0$, $\Omega\equiv1$ for $\delta\le\ub\le\delta+\varepsilon_0$, and integrating the equations for $D\tr\chi$, $D\eta$, $D\omegab$, $D\betab$, $D\alphab$ and Gauss-Codazzi-Ricci equations. The detail can also be found in \cite{L-Y} and we prefer to skip it here.

As a consequence, if $\delta$ is sufficiently small, the smooth solution $g$ of the vacuum Einstein equations exists in the region $\delta\le\ub\le\delta+\varepsilon_0$, $u_0\le u\le u_1$, and the solution $g$ is $\delta^{\frac{1}{2}}$-close to the Schwarzschild metric $g_{m_0}$ in the $C^k$ topology (for a smaller $k$). The proof of this statement can also be found in \cite{L-Y}. This is quite routine and we also prefer to skip it. This completes the proof of Theorem \ref{step1}.
\begin{remark}\label{Cauchystability}
This last statement is in fact a direct consequence of Cauchy stability, provided that not only the angular derivatives, but also the $D$, $\Db$ derivatives of the connection coefficients and curvature components are under control. From the null structure equations and null Bianchi equations, except $D\omega$, $\Db\omegab$, $D\alpha$ and $\Db\alphab$, all other $D$, $\Db$ derivatives of the connection coefficients and curvature components can be expressed in terms of the angular derivatives of some other components and and lower order terms. To control $D\omega$, $\Db\omegab$, $D\alpha$, $\Db\alphab$, we only need to commute $D$, $\Db$ with the null structure and Bianchi equations for $\Db\omega$, $D\omegab$, $\Db\alpha$ and $D\alphab$ and then integrate them. In summarize, the following estimates are true:
\begin{align*}
\|D^i\omega-(D^i\omega)_m, \Db^i\omegab-(\Db^i\omegab)_m, D^i\alpha, \Db^i\alphab\|_{H^j(\delta,u)}\le\delta^{\frac{1}{2}}C
\end{align*}
for $i+j=k$.
\end{remark}

\end{proof}

\begin{proof}[Proof of Theorem \ref{step3}]
By reversing time, Theorem \ref{existence} can also be applied for the case that the initial data are given on the future boundary of the spacetime region. Precisely, let us denote the norm defined on $C_{u_1}\cup\Cb_{\delta}$
$$\mathcal{R}_k^{(1)}=\sum_{R\in\{\alpha,\beta,\rho,\sigma,\betab\}}\mathcal{R}_k[R](u_1)+\sum_{\underline{R}\in\{\beta,\rho,\sigma,\betab,\alphab\}}\underline{\mathcal{R}}_k[\underline{R}](\delta).$$

\begin{align*}\mathcal{O}_k^{(1)}=\sup_{[\ub,u]\in \{\delta\}\times[u_0,u_1] \cup [0,\delta]\times\{u_1\}}\left(\sum_{\Gamma\in\{\chih, \tr\chi, \chibh, \widetilde{\tr\chib}, \eta, \etab, \omega, \omegab\}}\mathcal{O}_k[\Gamma](\ub,u)\right.\\
\left.+\delta^{-\frac{1}{2}}\|\log\Omega\|_{H^k(\ub,u)}+\delta^{-\frac{1}{2}}\left||u|^{-2}\gs_{AB}-\overset{\circ}{\gs}_{AB}\right|\right),\end{align*}
then Theorem \ref{existence} is still true if we replace the assumption $\mathcal{O}_k^{(0)},\mathcal{R}_k^{(0)}<\infty$ by
$$\mathcal{O}_k^{(1)},\mathcal{R}_k^{(1)}<\infty.$$
Theorem \ref{step3} follows by applying this reversed version of Theorem \ref{existence} for $u_0=u_0^*$, $u_1=u_1^*$.

\end{proof}

\begin{proof}[Proof of Theorem \ref{thm:closetoM}]
Note that the sphere $S_{0,u_1^*}$ is the boundary of the Minkowskian slice $\Sigma^+_{IV}$, then all connection coefficients and curvature components equal to their values in the Minkowski space. The conclusion of Proposition \ref{prop:smallness} still holds if $u_0$ is replaced by $u_1$ (in fact any $u\in[u_0,u_1]$) in \eqref{additioncondition}. By applying this reversed version of Proposition \ref{prop:smallness} for $m=0$, it follows that  if $\delta$ is sufficiently small, all connection coefficients and curvature components are $\delta^{\frac{1}{2}}$-close the their values in the Minkowski space in $H^k(0,u)$. Using the same method in Remark \ref{Cauchystability}, the $D$ and $\Db$ derivatives of all connection coefficients and curvature components are also close to their values in the Minkowski space. This would be sufficient for Cauchy stability and the proof of Theorem \ref{thm:closetoM} is completed.

\end{proof}

\section{Proof of Theorem \ref{thm:gluing}: A local deformation inside the black hole}\label{gluing}

Let us begin the proof by recalling the Kerr metric written in Boyer-Lindquist coordinates
\begin{align*}
g_{m,(0,0,a)}=&\left(-1+\frac{2mr}{\rho^2}\right)\D t^2-\frac{4mra\sin^2\theta}{\rho^2}\D t\D\varphi+\frac{\rho^2}{\Delta}\D r^2\\
&+\rho^2\D\theta^2+\sin^2\theta\left(r^2+a^2+\frac{2mra^2\sin^2\theta}{\rho^2}\right)\D\varphi^2,
\end{align*}
where $\rho^2=r^2+a^2\cos^2\theta$, $\Delta=r^2-2mr+a^2$. %We write $(m,(0,0,a))$ as a superscript in order the emphasize that $g^{m,(0,0,a)}$ is a covariant tensor. 
Let $r_+,r_-$ be the larger and smaller root of $\Delta=0$, then for any $r_0\in(r_-,r_+)$, the hypersurface $r=r_0$ is a spacelike hypersurface inside the black hole. When $a=0$, the Kerr metric reduces to Schwarzschild metric $g_{m}$. The Schwarzschild initial data $(\bar{g}_{m},\bar{k}_{m})$ induced on $r=r_0<2m$ writes
\begin{equation}\label{gmkm}
\begin{split}
\bar{g}_{m}&=\left(\frac{2m}{r_0}-1\right)\D t^2+r_0^2\D\theta^2+r_0^2\sin^2\theta\D\varphi^2,\\
\bar{k}_{m}&=mr_0^{-2}\left(\frac{2m}{r_0}-1\right)^{\frac{1}{2}}\D t^2-r_0\left(\frac{2m}{r_0}-1\right)^{\frac{1}{2}}\D \theta^2-r_0\sin^2\theta\left(\frac{2m}{r_0}-1\right)^{\frac{1}{2}}\D\varphi^2.
\end{split}
\end{equation}
We also need the expressions of the Kerr initial data  $(\bar{g}_{m,a},\bar{k}_{m,a})$. It is rather complicated, but fortunately, we do not need the exact expressions. What we need to know is the exact differences between Kerr and Schwarzschild up to $a$ to the first order. By direct computation,
\begin{equation}\label{gmakma}
\begin{split}
\bar{g}_{m,a}&=\bar{g}_{m}-4mar_0^{-1}\sin^2\theta\D t\D\varphi+O(a^2),\\
\bar{k}_{m,a}&=\bar{k}_{m}-2mar_0^{-2}\sin^2\theta\left(\frac{2m}{r_0}-1\right)^{\frac{1}{2}}\D t\D\varphi+O(a^2).
\end{split}
\end{equation}
For each initial data set $(\bar{g},\bar{k})$, we also use the momentum tensor
\begin{align*}
\bar{\pi}^{ij}=\bar{k}^{ij}-\tr_{\bar{g}}\bar{k}\bar{g}^{ij}
\end{align*}
instead of using $\bar{k}$ itself. {\bf $\bar{\pi}$ is understood as a two-contravariant tensor}.

\begin{comment}We also introduce the contravariant initial data $(\bar{g}_m,\bar{k}_m)$. We write $m$ as a subscript in order to emphasize that $\bar{g}_m,\bar{k}_m$ are contravariant tensors. 
Because $\bar{g}^m$ and $\bar{k}^m$ are diagonal, it is easy to see that 
\begin{align*}\bar{g}_m^{ij}&=(\bar{g}^m_{ij})^{-1},\\
k_m^{11}&=mr^{-2}\left(\frac{2m}{r}-1\right)^{-\frac{3}{2}},\\
k_m^{22}&=-r^{-3}\left(\frac{2m}{r}-1\right)^{\frac{1}{2}},\\
k_m^{33}&=-r^{-3}(\sin\theta)^{-2}\left(\frac{2m}{r}-1\right)^{\frac{1}{2}}.\end{align*}
We also need the trace of $k_m$:
$$\tr_{\bar{g}^{m}}k_{m}=r^{-2}\left(\frac{2m}{r}-1\right)^{-\frac{1}{2}}(2r-3m).$$
We also have the contravariant initial data: 
\begin{align*}
\bar{g}_{m,(0,0,a)}^{ii}&=\bar{g}_{m}^{ii}+O(a^2),\\
\bar{g}_{m,(0,0,a)}^{13}=\bar{g}_{m,(0,0,a)}^{31}&=2mar^{-3}\left(\frac{2m}{r}-1\right)^{-1}+O(a^2),\\
\bar{k}_{m,(0,0,a)}^{ii}&=\bar{k}_{m}^{ii}+O(a^2),\\
\bar{k}_{m,(0,0,a)}^{13}=\bar{k}_{m,(0,0,a)}^{31}&=mar^{-4}\left(\frac{2m}{r}-1\right)^{-\frac{3}{2}}\left(3-\frac{4m}{r}\right)+O(a^2).
\end{align*}
{\bf Here the indices  are raised using $\bar{g}_{m,(0,0,a)}$}. 
\end{comment}

We then introduce the 4-parameter Kerr initial data family. In the rectangular coordinate $x_1=r\cos\varphi\sin\theta$, $x_2=r\sin\varphi\sin\theta$, $x_3=r\cos\theta$, for an arbitrary vector $\mathbf{a}=(a_1,a_2,a_3)\in\mathbb{R}^3$ with $|\mathbf{a}|=a$  and an isometry $\Omega_{\mathbf{a}}\in SO(3)$ mapping $\mathbf{a}$ to $(0,0,a)$, define the Kerr metric $g_{m,\mathbf{a}}=(\mathrm{id}_{\mathbb{R}}\times\Omega_{\mathbf{a}})^* \,g_{m,a}$ where $\mathrm{id}_{\mathbb{R}}$ is the identity map of the $t$ axis.  Because of the axial symmetry of the Kerr metric,  this definition does not depend on the choice of $\Omega_{\mathbf{a}}$.  We have in particular $g_{m,(0,0,a)}=g_{m,a}$. We would like to prove the following proposition and Theorem \ref{thm:gluing} follows immediately.

\begin{proposition}\label{gluinginside}
Let $H\cong(t_1,t_2)\times \mathbb{S}^2$ and $r_0\in (0,2m_0)$ where $m_0>0$.  Let $(\bar{g}_{m_0},\bar{\pi}_{m_0})$ be the initial data defined on $H$ that is equal to the initial data induced on the slice $r=r_0, t\in(t_1,t_2)$ of the Schwarzschild metric. Suppose that an initial data set $(\bar{g},\bar{\pi})$ defined on $H$ satisfies
\begin{align*}\|(\bar{g},\bar{\pi})-(\bar{g}_{m_0},\bar{\pi}_{m_0})\|_{C^{k,\alpha}(\bar{g}_{m_0})}<\varepsilon
\end{align*}
 for some sufficiently large integer $k$ and $\alpha\in (0,1)$. If $\varepsilon$ is sufficiently small, then there exists another initial data set $(\tilde{g},\tilde{\pi})$ defined on $H$ such that near $t=t_1$,
  \begin{align*}(\tilde{g},\tilde{\pi})=(\bar{g},\bar{\pi})
  \end{align*}
  and near $t=t_2$,
  \begin{align*}(\tilde{g},\tilde{\pi})=(\bar{g}_{m,\bf{a}},\bar{\pi}_{m,\bf{a}})\end{align*} where $(\bar{g}_{m,\bf{a}},\bar{\pi}_{m,\bf{a}})$ equals to the Kerr initial data induced on the slice $r=r_0$ inside the black hole with parameters $(m,\bf{a})$. Moreover, 
  \begin{align*}\|(\tilde{g},\tilde{\pi})-(\bar{g},\bar{\pi})\|_{C^{k,\alpha}(\bar{g}_{m_0})}< C\varepsilon.\end{align*}
\end{proposition}

\begin{proof}
Let $\phi=\phi(t)$ be a smooth cut-off function on $H$ such that $\phi=1$ near $t=t_1$ and $\phi = 0$ near $t=t_2$. We define an approximation to the final initial data set as follows:
\begin{align*}
(\tilde{g},\tilde{\pi})= \big(\phi\bar{g}+(1-\phi)\bar{g}_{m,\mathbf{a}},\phi\bar{\pi} + (1-\phi)\bar{\pi}_{m,\mathbf{a}}\big).
\end{align*}
Here $(m,\bf{a})$ is chosen later such that $|m-m_0|+|{\bf{a}}|\le C_0\varepsilon$ for some $C_0$. It is direct to see $(\tilde{g},\tilde{k})$ equals to $(\bar{g},\bar{\pi})$ near $t=t_1$ and equals to $(\bar{g}_{m,\bf{a}},\bar{\pi}_{m,\bf{a}})$ near $t=t_2$. However, $(\tilde{g},\tilde{k})$ does not necessarily satisfy the constraint equations $\Phi(\tilde{g},\tilde{\pi})=0$, where $\Phi$ is the constraint map defined by
$$\Phi(\bar{g},\bar{\pi})=\left(R(\bar{g})+\frac{1}{2}\tr_{\bar{g}}\bar{\pi}-|\bar{\pi}|^2_{\bar{g}}, \mathrm{div}_{\bar{g}}\bar{\pi}\right),$$
where $R(\bar{g})$ is the scalar curvature of $\bar{g}$.
On the other hand, by the closeness of $(\bar{g},\bar{\pi})$ and $(\bar{g}_{m_0},\bar{\pi}_{m_0})$, $\Phi(\tilde{g},\tilde{\pi})$ is small in some weighted H\"older norm and moreover, $\Phi(\tilde{g},\tilde{\pi})$ is compactly supported in $H$.

By local deformation techniques by Corvino-Schoen \cite{C-S}, if $\varepsilon$ is sufficiently small, there exists an $(h,\omega)$ which is compactly supported in $H$ such that 
\begin{align*}\Phi(\tilde{g}+h,\tilde{\pi}+\omega)\in\zeta\cdot \mathrm{Ker}D\Phi_{(\bar{g}_{m_0},\bar{\pi}_{m_0})}^*\end{align*} where $\zeta=\zeta(t)$ is a given bump function compactly supported in $H$, $\mathrm{Ker}D\Phi_{(\bar{g}_{m_0},\bar{\pi}_{m_0})}^*$ is the kernel of the formal $L^2$-adjoint of $D\Phi_{(\bar{g}_{m_0},\bar{\pi}_{m_0})}$, the differential of the constraint map at the Schwarzschildean data set $(\bar{g}_{m_0},\bar{\pi}_{m_0})$. Moreover,
\begin{align*}
\|(h,\omega)\|_{C^{k,\alpha}(\bar{g}_{m_0})}< C\varepsilon.
\end{align*}
As shown in \cite{M}, for any given initial data set $(\bar{g},\bar{\pi})$ on a Cauchy hypersurface, a nontrivial element in $\mathrm{Ker}D\Phi_{(\bar{g},\bar{\pi})}^*$ corresponds to a Killing field of the corresponding vacuum solution, and conversely, a nontrivial element in $\mathrm{Ker}D\Phi_{(\bar{g},\bar{\pi})}^*$ can be obtained by projecting a Killing field of the vacuum solution to the normal and tangential directions of the Cauchy hypersurface. The Killing fields of the Schwarzschild metric $g_{m_0}$ are $\partial_t$ and three rotation vectorfields. Because they are tangent to the hypersurface $r=r_0<2m_0$, so they also span $\mathrm{Ker}D\Phi_{(\bar{g}_{m_0},\bar{\pi}_{m_0})}^*$. More precisely, the $4$-dimensional linear space $\mathrm{Ker}D\Phi_{(\bar{g}_{m_0},\bar{\pi}_{m_0})}^*$ is spanned by 
\begin{align*}(0,\partial_t), (0,\Omega_1),(0,\Omega_2),(0,\Omega_3)\end{align*} 
where $$\Omega_1=-\sin\varphi\partial_\theta-\cos\varphi\cot\theta\partial_\varphi, \Omega_2=\cos\varphi\partial_\theta-\sin\varphi\cot\theta\partial_\varphi, \Omega_3=\partial_\varphi$$
are the rotation Killing fields written in spherical coordinates.

The remaining thing is to show that we can suitably choose $(m,{\bf a})$ such that $\Phi(\tilde{g}+h,\tilde{\pi}+\omega)=0$. It is easy to verify that 
\begin{align*}L^2(H,\bar{g}_{m_0})=\mathrm{Ker}D\Phi_{(\bar{g}_{m_0},\bar{\pi}_{m_0})}\oplus(\zeta\cdot \mathrm{Ker}D\Phi_{(\bar{g}_{m_0},\bar{\pi}_{m_0})}^*)^\perp.\end{align*} So we need to verify that $\Phi(\tilde{g},\tilde{\pi})$ has no components in $\mathrm{Ker}D\Phi_{(\bar{g}_{m_0},\bar{\pi}_{m_0})}^*$. Denoting $\Omega_0=\partial_t$, $\Omega_\alpha, \alpha=0,1,2,3$, we want to choose $(m,{\bf{a}})$ that is a zero of the following maps
\begin{align}\label{Ialphama}
\mathcal{I}_\alpha(m,a_1,a_2,a_3)=(\Phi(\tilde{g}+h,\tilde{\pi}+\omega),(0,\Omega_\alpha))_{L^2(\bar{g}_{m_0})}=\int_{H}(\mathrm{div}_{\tilde{g}+h}(\tilde{\pi}+\omega))^i(\Omega_{\alpha})_i\D\mu_{\bar{g}_{m_0}}
\end{align}
for $\alpha=0,1,2,3$. We will find the exact expressions of $\mathcal{I}_\alpha$ up to $\varepsilon$ to the first order. To this end, we expand $\mathrm{div}_{\tilde{g}+h}{(\tilde{\pi}+\omega)}$ at $(\bar{g}_{m_0},\bar{\pi}_{m_0})$ as the following (see \cite{FM}):
\begin{equation}\label{taylordiv}
\begin{split}
(\mathrm{div}_{\tilde{g}+h}(\tilde{\pi}+\omega))^i=
&(D\mathrm{div}_{(\bar{g}_{m_0},\bar{\pi}_{m_0})}(\tilde{g}-\bar{g}_{m_0}+h,\tilde{\pi}-\bar{\pi}_{m_0}+\omega))^i+O(\varepsilon^2)\\
=&(\mathrm{div}_{\bar{g}_{m_0}}(\tilde{\pi}-\bar{\pi}_{m_0}+\omega))^i-\frac{1}{2}\bar{\pi}_{m_0}^{jk}(\tilde{g}-\bar{g}_{m_0}+h)_{jk;}^{\phantom{aaa}i}\\&+\bar{\pi}_{m_0}^{jk}(\tilde{g}-\bar{g}_{m_0}+h)^{i}_{\phantom{a}j;k}+\frac{1}{2}\bar{\pi}_{m_0}^{ij}(\tr_{\bar{g}_{m_0}}(\tilde{g}-\bar{g}_{m_0}+h))_{;j}+O(\varepsilon^2)
\end{split}
\end{equation}
%$$(D\mathrm{div}_{(\widetilde{g},\widetilde{\pi})}(h,\omega))^i=(\mathrm{div}_{\widetilde{g}}\omega)^i+\frac{1}{2}\widetilde{\pi}^{jk}h_{jk;}^{\phantom{aaa}i}-\widetilde{\pi}^{jk}h^{i}_{\phantom{a}j;k}-\frac{1}{2}\widetilde{\pi}^{ij}(\tr_{\widetilde{g}}h)_{;j}.$$
where the indices are raised and lowed by $\bar{g}_{m_0}$ and the semicolons means taking covariant derivative relative to $\bar{g}_{m_0}$. We are going to compute the integral \eqref{Ialphama} by evaluating the integral of the four terms in \eqref{taylordiv} over $H$. Because $\Omega_\alpha$ are Killing, i.e., $(\Omega_\alpha)_{i;j}$ is antisymmetric, by divergence theorem, for the first term, we have 
\begin{align*}
\int_H(\mathrm{div}_{\bar{g}_{m_0}}(\tilde{\pi}-\bar{\pi}_{m_0}+\omega))^i(\Omega_\alpha)_i\D\mu_{\bar{g}_{m_0}}=\int_{t=t_2}+\int_{t=t_1}(\tilde{\pi}-\bar{\pi}_{m_0}+\omega)^{ij}(\Omega_{\alpha})_{i}N_j
\end{align*}
where $N=\left(\frac{2m_0}{r_0}-1\right)^{-\frac{1}{2}}\partial_t$ is the unit normal of the boundary of $H$. $N$ is pointing outward at $t=t_2$ and inward at $t=t_1$. For the second term, we have
\begin{align*}
\int_H-\frac{1}{2}\bar{\pi}_{m_0}^{jk}(\tilde{g}-\bar{g}_{m_0}+h)_{jk;}^{\phantom{aaa}i}(\Omega_\alpha)_i\D\mu_{\bar{g}_{m_0}}=\int_{t=t_2}+\int_{t=t_1}-\frac{1}{2}\bar{\pi}_{m_0}^{jk}(\tilde{g}-\bar{g}_{m_0}+h)_{jk}(\Omega_\alpha)_iN^i
\end{align*}
where we have used in addition $(\bar{\pi}_{m_0}^{jk})_;^{\phantom{;}i}=0$ which can be checked by direct computation using \eqref{gmkm}. For the fourth term, in a similar way, we have
\begin{align*}
&\int_H\frac{1}{2}\bar{\pi}_{m_0}^{ij}(\tr_{\bar{g}_{m_0}}(\tilde{g}-\bar{g}_{m_0}+h))_{;j}(\Omega_\alpha)_i\D\mu_{\bar{g}_{m_0}}\\=&\int_{t=t_2}+\int_{t=t_1}\frac{1}{2}\bar{\pi}_{m_0}^{ij}(\tr_{\bar{g}_{m_0}}(\tilde{g}-\bar{g}_{m_0}+h))(\Omega_\alpha)_iN_j.
\end{align*}
For the third term, the situation is different. We have
\begin{align*}
\int_H\bar{\pi}_{m_0}^{jk}(\tilde{g}-\bar{g}_{m_0}+h)^i_{\phantom{a}j;k}(\Omega_\alpha)_i\D\mu_{\bar{g}_{m_0}}=&\int_{t=t_2}+\int_{t=t_1}\bar{\pi}_{m_0}^{jk}(\tilde{g}-\bar{g}_{m_0}+h)^i_{\phantom{a}j}(\Omega_\alpha)_iN_k\\
&+\int_H -\bar{\pi}_{m_0}^{jk}(\tilde{g}-\bar{g}_{m_0}+h)^i_{\phantom{a}j}(\Omega_\alpha)_{i;k}\D\mu_{\bar{g}_{m_0}}.
\end{align*}
We claim that the integrant in the second line vanishes. In fact, for $\alpha=0$, it is clear that $(\partial_t)_{i;k}=0$. For $\alpha=1,2,3$, it still holds that $\bar{\pi}_{m_0}^{jk}(\Omega_\alpha)^{i}_{\phantom{i};k}$ is antisymmetric relative to $i,j$ by direct computation using \eqref{gmkm}. Combining all the above formulas, we have
\begin{equation}\label{Ialphama1}
\begin{split}
&\mathcal{I}_\alpha(m,a_1,a_2,a_3)=\int_{t=t_2}(\tilde{\pi}-\bar{\pi}_{m_0}+\omega)^{ij}(\Omega_{\alpha})_{i}N_j-\int_{t=t_2}\frac{1}{2}\bar{\pi}_{m_0}^{jk}(\tilde{g}-\bar{g}_{m_0}+h)_{jk}(\Omega_\alpha)_iN^i\\
&+\int_{t=t_2}\bar{\pi}_{m_0}^{jk}(\tilde{g}-\bar{g}_{m_0}+h)^i_{\phantom{a}j}(\Omega_\alpha)_iN_k+\int_{t=t_2}\frac{1}{2}\bar{\pi}_{m_0}^{ij}(\tr_{\bar{g}_{m_0}}(\tilde{g}-\bar{g}_{m_0}+h))(\Omega_\alpha)_iN_j\\&+\epsilon_\alpha+O(\varepsilon^2)
\end{split}
\end{equation}
where $\epsilon_\alpha$ equals to the integral of the same integrants over $t=t_1$. We have $|\epsilon_\alpha|\le C\varepsilon$ and most importantly, because $(h,w)$ has compact support in $H$,  $\epsilon_\alpha$ does not depend on the choice of $(m,a_1,a_2,a_3)$.

We first compute $\mathcal{I}_\alpha(m,0,0,a)$. Noting that in this case \begin{align*}(\tilde{g},\tilde{\pi})=(\bar{g}_{m,a},\bar{\pi}_{m,a})\end{align*} at $t=t_2$ and $(h,\omega)$ has compact support, from \eqref{Ialphama1}, we have
\begin{equation}\label{Ialpha}
\begin{split}
&\mathcal{I}_\alpha(m,0,0,a)=\int_{t=t_2}(\bar{\pi}_{m,a}-\bar{\pi}_{m_0})^{ij}(\Omega_\alpha)_iN_j-\int_{t=t_2}\frac{1}{2}\bar{\pi}_{m_0}^{jk}(\bar{g}_{m,a}-\bar{g}_{m_0})_{jk}(\Omega_\alpha)_iN^i\\
&+\int_{t=t_2}\bar{\pi}_{m_0}^{jk}(\bar{g}_{m,a}-\bar{g}_{m_0})^{i}_{\phantom{a}j}(\Omega_\alpha)_iN_k+\int_{t=t_2}\frac{1}{2}\bar{\pi}_{m_0}^{ij}(\tr_{\bar{g}_{m_0}}(\bar{g}_{m,a}-\bar{g}_{m_0}))(\Omega_\alpha)_iN_j\\
&+\epsilon_\alpha+O(\varepsilon^2).
\end{split}
\end{equation}
To compute $\mathcal{I}_0(m,0,0,a)$, we compute four terms of \eqref{Ialpha} for $\alpha=0$ using \eqref{gmkm} and \eqref{gmakma}: The first term is
\begin{align*}
(\bar{\pi}_{m,a}-\bar{\pi}_{m_0})^{ij}(\Omega_0)_iN_j=&(\bar{\pi}_{m,a}^{ij}-\bar{\pi}_{m_0}^{ij})(\partial_t)_i\cdot\left(\frac{2m_0}{r_0}-1\right)^{-\frac{1}{2}}(\partial_t)_j\\=&(\bar{\pi}_{m}^{11}-\bar{\pi}_{m_0}^{11})\left(\frac{2m_0}{r_0}-1\right)^\frac{3}{2}+O(\varepsilon^2),
\end{align*}
where we have used
$$(\partial_t)_1=(\bar{g}_{m_0})_{11}(\partial_t)^1=\frac{2m_0}{r_0}-1, (\partial_t)_2=(\partial_t)_3=0.$$
the second term is
\begin{align*}
-\frac{1}{2}\bar{\pi}_{m_0}^{jk}(\bar{g}_{m,a}-\bar{g}_{m_0})_{jk}(\Omega_0)_iN^i
=&-\frac{1}{2}\bar{\pi}_{m_0}^{ij}(\bar{g}_{m,a}-\bar{g}_{m_0})_{ij}\cdot\left(\frac{2m_0}{r_0}-1\right)^{-\frac{1}{2}}|\partial_t|^2\\=&-\frac{1}{2}\bar{\pi}_{m_0}^{11}2(m-m_0)r_0^{-1}\left(\frac{2m_0}{r_0}-1\right)^{\frac{1}{2}}+O(\varepsilon^2),
\end{align*}
the third term is
\begin{align*}
\bar{\pi}_{m_0}^{jk}(\bar{g}_{m,a}-\bar{g}_{m_0})^{i}_{\phantom{a}j}(\Omega_0)_iN_k
=&\bar{\pi}_{m_0}^{jk}\bar{g}_{m_0}^{il}(\bar{g}_{m,a}-\bar{g}_{m_0})_{lj}(\partial_t)_k\cdot\left(\frac{2m_0}{r_0}-1\right)^{-\frac{1}{2}}(\partial_t)_i\\=&\bar{\pi}_{m_0}^{11}2(m-m_0)r_0^{-1}\left(\frac{2m_0}{r_0}-1\right)^{\frac{1}{2}}+O(\varepsilon^2),\\
\end{align*}
and the fourth term is
\begin{align*}
&\frac{1}{2}\bar{\pi}_{m_0}^{ij}(\tr_{\bar{g}_{m_0}}(\bar{g}_{m,a}-\bar{g}_{m_0}))(\Omega_0)_iN_j
\\=&\frac{1}{2}\bar{\pi}_{m_0}^{ij}\bar{g}_{m_0}^{kl}(\bar{g}_{m,a}-\bar{g}_{m_0})_{kl}(\partial_t)_i\cdot\left(\frac{2m_0}{r_0}-1\right)^{-\frac{1}{2}}(\partial_t)_j\\=&\frac{1}{2}\bar{\pi}_{m_0}^{11}2(m-m_0)r_0^{-1}\left(\frac{2m_0}{r_0}-1\right)^\frac{1}{2}+O(\varepsilon^2).
\end{align*}
Plugging in
$$\bar{\pi}_{m_0}^{11}=2r_0^{-1}\left(\frac{2m_0}{r_0}-1\right)^{-\frac{1}{2}}$$
and
$$\bar{\pi}_{m}^{11}-\bar{\pi}_{m_0}^{11}=-2r_0^{-2}\left(\frac{2m_0}{r_0}-1\right)^{-\frac{3}{2}}(m-m_0)+O(\varepsilon^2)$$
to the above formulas,  we then have
\begin{align*}
\mathcal{I}_0(m,0,0,a)=&2r_0^{-2}(m-m_0)\cdot 4\pi r_0^2+\epsilon_0+O(\varepsilon^2)\\
=&8\pi(m-m_0)+\epsilon_0+O(\varepsilon^2).
\end{align*}
The crucial fact is that the coefficient of $(m-m_0)$ is not zero.

We then compute \eqref{Ialpha} for $\alpha=1,2,3$. Recall again that the rotation vectorfields are $$\Omega_1=-\sin\varphi\partial_\theta-\cos\varphi\cot\theta\partial_\varphi, \Omega_2=\cos\varphi\partial_\theta-\sin\varphi\cot\theta\partial_\varphi, \Omega_3=\partial_\varphi.$$
We firstly compute $\mathcal{I}_3(m,0,0,a)$ and we will see the others can be obtained through symmetries. Similar to computing $\mathcal{I}_0$, we compute \eqref{Ialpha} for $\alpha=3$: The first term is
\begin{align*}
(\bar{\pi}_{m,a}-\bar{\pi}_{m_0})^{ij}(\Omega_3)_iN_j=&(\bar{\pi}_{m,a}-\bar{\pi}_{m_0})^{ij}(\partial_\varphi)_i\cdot\left(\frac{2m_0}{r_0}-1\right)^{-\frac{1}{2}}(\partial_t)_j\\=&\bar{\pi}_{m,a}^{13}\left(\frac{2m_0}{r_0}-1\right)^{\frac{1}{2}}r_0^2\sin^2\theta+O(\varepsilon^2),
\end{align*}
where we have used
$$(\partial_\varphi)_1=(\partial_\varphi)_2=0,(\partial_\varphi)_3=(\bar{g}_{m_0})_{33}(\partial_\varphi)^3=r^2_0\sin^2\theta.$$
the second term vanishes because $(\Omega_3)_iN^i=0$. The third term is
\begin{align*}
\bar{\pi}_{m_0}^{jk}(\bar{g}_{m,a}-\bar{g}_{m_0})^{i}_{\phantom{a}j}(\Omega_3)_iN_k
=&\bar{\pi}_{m_0}^{jk}\bar{g}_{m_0}^{il}(\bar{g}_{m,a}-\bar{g}_{m_0})_{lj}\cdot\left(\frac{2m_0}{r_0}-1\right)^{-\frac{1}{2}}(\partial_t)_k(\partial_\varphi)_i\\=&\bar{\pi}_{m_0}^{11}\bar{g}_{m_0}^{33}(\bar{g}_{m,a})_{31}\left(\frac{2m_0}{r_0}-1\right)^{\frac{1}{2}}r_0^2\sin^2\theta+O(\varepsilon^2),
\end{align*}
and the fourth term also vanishes because $\bar{\pi}_{m_0}^{ij}(\Omega_3)_iN_j=0$.
Using the formula%$\pi_{m_0}^{33}=\left(\frac{2m_0}{r}-1\right)^{-1/2}r^{-4}\sin^{-2}\theta(m_0-r)$, 
 $$\bar{\pi}_{m,a}^{13}=mar_0^{-4}\left(\frac{2m}{r_0}-1\right)^{-\frac{1}{2}}+O(\varepsilon^2),$$ 
we will have
\begin{align*}
\mathcal{I}_3(m,0,0,a)=\int_0^{2\pi}\D\varphi\int_0^{\pi}-3m_0a\sin^3\theta\D\theta+\epsilon_3+O(\varepsilon^2)=-8\pi m_0a+\epsilon_3+O(\varepsilon^2).
%\frac{8}{3}\pi\cdot m_0ar^{-2}(2m_0r^{-1}-1)^{-1/2}(4m_0r^{-1}-3)+O(\varepsilon^2).
\end{align*}
The crucial fact is again that the coefficient of $a$ is not zero. For $\alpha=1,2$, note that in the expression of $\Omega_1,\Omega_2$, the coefficients of both $\partial_\theta$ and $\partial_\varphi$ have a factor $\sin\varphi$ or $\cos\varphi$, whose integrals over $0$ to $2\pi$ is zero. Therefore, it is direct to see that 
$$\mathcal{I}_1(m,0,0,a)=\epsilon_1+O(\varepsilon^2),\ \mathcal{I}_2(m,0,0,a)=\epsilon_2+O(\varepsilon^2).$$

We then compute $\mathcal{I}_\alpha(m,a,0,0)$. Similar to \eqref{Ialpha}, we have
\begin{equation*}
\begin{split}
&\mathcal{I}_\alpha(m,a,0,0)=\int_{t=t_2}(\bar{\pi}_{m,(a,0,0)}-\bar{\pi}_{m_0})^{ij}(\Omega_\alpha)_iN_j-\int_{t=t_2}\frac{1}{2}\bar{\pi}_{m_0}^{jk}({\bar{g}}_{m,(a,0,0)}-\bar{g}_{m_0})_{jk}(\Omega_\alpha)_iN^i\\
&+\int_{t=t_2}\bar{\pi}_{m_0}^{jk}(\bar{g}_{m,(a,0,0)}-\bar{g}_{m_0})^{i}_{\phantom{a}j}(\Omega_\alpha)_iN_k+\int_{t=t_2}\frac{1}{2}\bar{\pi}_{m_0}^{ij}(\tr_{\bar{g}_{m_0}}(\bar{g}_{m,(a,0,0)}-\bar{g}_{m_0}))(\Omega_\alpha)_iN_j\\
&+\epsilon_\alpha+O(\varepsilon^2).
\end{split}
\end{equation*}
Let us choose $R\in SO(3)$:
$$R=\begin{pmatrix}0 & 0&-1\\0&1&0\\1&0&0\end{pmatrix}.$$
Then $R$ maps $(1,0,0)$ to $(0,0,1)$. Then $R^*(\bar{g}_{m,(0,0,a)},\bar{\pi}_{m(0,0,a)})=(\bar{g}_{m,(a,0,0)},\bar{\pi}_{m(a,0,0)})$. For $\alpha=0$, recall that $\Omega_0=\partial_t$ and note that $\D R(\partial_t)=\partial_t$, we find that in the four terms above, both $\bar{g}_{m,(a,0,0)}$ and $\bar{\pi}_{m(a,0,0)}$ contract with factors that are invariant under the pullback or pushforward induced by  $R$. Therefore, we can conclude that
$$\mathcal{I}_0(m,a,0,0)=\mathcal{I}_0(m,0,0,a)+O(\varepsilon^2)=8\pi(m-m_0)+\epsilon_0+O(\varepsilon^2).$$
Similarly, we also have
$$\mathcal{I}_0(m,0,a,0)=\mathcal{I}_0(m,0,0,a)+O(\varepsilon^2)=8\pi(m-m_0)+\epsilon_0+O(\varepsilon^2).$$
For $\alpha=1,2,3$, note that $\D R$ maps $(\Omega_1,\Omega_2,\Omega_3)$ to $(\Omega_3,\Omega_2,-\Omega_1)$. We have
\begin{align*}(\bar{\pi}_{m,(a,0,0)}-\bar{\pi}_{m_0})^{ij}(\Omega_\alpha)_iN_j=&(R^*(\bar{\pi}_{m,a}-\bar{\pi}_{m_0}))^{ij}(\Omega_\alpha)_iN_j\\=&(\bar{\pi}_{m,a}-\bar{\pi}_{m_0})^{ij}(\D R(\Omega_\alpha))_iN_j
\end{align*}
and similarly,
$$\bar{\pi}_{m_0}^{jk}(\bar{g}_{m,(a,0,0)}-\bar{g}_{m_0})^{i}_{\phantom{a}j}(\Omega_\alpha)_iN_k=\bar{\pi}_{m_0}^{jk}(\bar{g}_{m,a}-\bar{g}_{m_0})^{i}_{\phantom{a}j}(\D R(\Omega_\alpha))_iN_k.$$
Finally, the second and the fourth terms vanish. Therefore, we have
$$\mathcal{I}_1(m,a,0,0)=\mathcal{I}_3(m,0,0,a)-\epsilon_3+\epsilon_1+O(\varepsilon^2)=-8\pi m_0a+\epsilon_1+O(\varepsilon^2)$$
and
$$\mathcal{I}_2(m,a,0,0)=\epsilon_2+O(\varepsilon^2),\ \mathcal{I}_3(m,a,0,0)=\epsilon_3+O(\varepsilon^2).$$

Finally, taking
$$R=\begin{pmatrix}1&0&0\\0 & 0&-1\\0&1&0\end{pmatrix}$$
which maps $(0,1,0)$ to $(0,0,1)$, we can compute similarly
$$\mathcal{I}_2(m,0,a,0)=\mathcal{I}_3(m,0,0,a)-\epsilon_3+\epsilon_2+O(\varepsilon^2)=-8\pi m_0a+\epsilon_2+O(\varepsilon^2)$$
and
$$\mathcal{I}_1(m,0,a,0)=\epsilon_1+O(\varepsilon^2),\ \mathcal{I}_3(m,0,a,0)=\epsilon_3+O(\varepsilon^2).$$

Using the above results and by Taylor expansion, we have
\begin{align*}\mathcal{I}(m,\mathbf{a})=&(\mathcal{I}_\alpha(m,a_1,a_2,a_3))_{\alpha=0,1,2,3}\\
=&(8\pi(m-m_0),-8\pi m_0\mathbf{a})+(\epsilon_0,\epsilon_1,\epsilon_2,\epsilon_3)+O(\varepsilon^2).
\end{align*}
We then use a degree argument similar to \cite{L-Y}. More precisely, choose $C_0$ such that 
\begin{align*}
\left(-\frac{\epsilon_0}{8\pi}+m_0, \frac{1}{8\pi m_0}(\epsilon_1,\epsilon_2,\epsilon_3)\right)\in B_{C_0\varepsilon}=\{(m,\mathbf{a}):|m-m_0|+|\mathbf{a}|\le C_0\varepsilon\}.
\end{align*}
Let us introduce
\begin{align*}\mathcal{I}(m,\mathbf{a},t)=(8\pi(m-m_0),-8\pi m_0\mathbf{a})+(\epsilon_0,\epsilon_1,\epsilon_2,\epsilon_3)+tO(\varepsilon^2).
\end{align*}
which satisfies $\mathcal{I}(m,\mathbf{a},1)=\mathcal{I}(m,\mathbf{a})$ and 
$$\mathcal{I}\left(-\frac{\epsilon_0}{8\pi}+m_0, \frac{1}{8\pi m_0}(\epsilon_1,\epsilon_2,\epsilon_3),0\right)=0.$$
So $\mathcal{I}(m,\mathbf{a},0)$ is a homeomorphism from $B_{C_0\varepsilon}$ to another box in $\mathbb{R}^4$ containing $0\in\mathbb{R}^4$. If $\varepsilon$ is sufficiently small, $0\notin \mathcal{I}(\partial B_{C_0\varepsilon}\times[0,1])$, and hence the degree of $\mathcal{I}$ at $0\in\mathbb{R}^4$ does not change for $t\in[0,1]$. This in particular implies that $\mathcal{I}(m,\mathbf{a})$ has a zero in $B_{C_0\varepsilon}$. This completes the proof of Proposition \ref{gluinginside}.

\end{proof}

Theorem \ref{thm:gluing} then follows by applying Proposition \ref{gluinginside} for $\varepsilon=C\delta^{\frac{1}{2}}$ (with a different $k$).

\end{document}